\documentclass[a4paper,11pt]{article}

\usepackage{latexsym}   

\usepackage{amsmath,amssymb,amsthm}
\usepackage{dsfont}
\usepackage{graphicx}
\usepackage[british]{babel}
\usepackage[font=small,labelfont=bf]{caption}
\usepackage[utf8]{inputenc}
\usepackage{float}
\usepackage{xcolor}
\usepackage[colorlinks=true,linkcolor=blue,citecolor=red]{hyperref}
\usepackage{authblk}
\usepackage{blindtext}
\usepackage{epstopdf}
\usepackage{caption}
\usepackage{subcaption}

\newtheorem{remark}{Remark}
\newcommand{\R}{\mathbb{R}}
\newcommand{\be}{\begin{equation}}
\newcommand{\ee}{\end{equation}}
\newcommand{\benn}{\begin{equation*}}
\newcommand{\eenn}{\end{equation*}}
\usepackage[margin=1cm]{caption}

\theoremstyle{plain}

\newtheorem{theorem}{Theorem}[section]

\renewcommand{\epsilon}{\varepsilon}
\begin{document}
\title{Effects of vaccination efficacy on wealth distribution\\ in kinetic epidemic models}

\author[1]{E. Bernardi \thanks{\tt emanuele.bernardi01@universitadipavia.it}}
\author[3]{L. Pareschi \thanks{\tt lorenzo.pareschi@unife.it}}
\author[1,2]{G. Toscani \thanks{\tt giuseppe.toscani@unipv.it}}
\author[1]{M. Zanella\thanks{\tt mattia.zanella@unipv.it}}

\affil[1]{Department of Mathematics "F. Casorati", University of Pavia, Italy}
\affil[2]{IMATI "E. Magenes", CNR, Pavia, Italy }
\affil[3]{Department of Mathematics and Computer Science, University of Ferrara, Italy}

\date{\today}
\maketitle

\abstract{The spreading of Covid-19 pandemic has highlighted the close link between economics and health in the context of emergency management. A widespread vaccination campaign is considered the main tool to contain the economic consequences. This paper will focus, at the level of wealth distribution modelling, on the economic improvements induced by the vaccination campaign in terms of its effectiveness rate. The economic trend during the pandemic is evaluated resorting to  a mathematical model joining a classical compartmental model including vaccinated individuals with a kinetic model of wealth distribution based on binary wealth exchanges. The interplay between wealth exchanges and the progress of the infectious disease is realized by assuming on the one hand that individuals in different compartments act differently in the economic process and on the other hand that the epidemic affects risk in economic transactions. Using the mathematical tools of kinetic theory, it is possible to identify the equilibrium states of the system and the formation of inequalities due to the pandemic in the wealth distribution of the population.
Numerical experiments highlight the importance of the vaccination campaign and its positive effects in reducing economic inequalities in the multi-agent society.}


\section{Introduction}

In the early 2020s, the spread of the Covid-19 pandemic highlighted the close link between economics and health in the context of emergency management.  Because of this, assessing the impact of an epidemic phenomenon on a country's economy has emerged as one of the key aspects to consider in the context of containment strategies. From a mathematical point of view,  a systematic approach to the study of the effects on the economies of countries facing a severe pandemic is a very complex problem and a mathematical model can only provide rough indications of the possible consequences, based on simplifying assumptions about the key parameters driving the pandemic evolution.
The basic idea is to trace these phenomena back to the evolution of the so-called wealth distribution of a country, which measures how many people belong to increasing income levels. 

A first attempt to understand changes in wealth distribution  in the presence of epidemic spread was proposed in \cite{DPTZ} by combining the classical SIR compartmental model of susceptible, infected and recovered individuals \cite{H} with the kinetic model of wealth distribution introduced in \cite{CPT}, and assuming that, due to the presence of the pandemic, individuals in different compartments act differently in the economic process. Although the model was developed in a relatively simplified context it has provided a general framework for socio-epidemiological modeling that can be easily extended to more complex dynamics both in terms of economic transactions \cite{CC1,Furioli} and in terms of epidemic interactions \cite{Gatto, Par21}. We mention in this direction the recent survey \cite{Albi_etal} and the seminal approach proposed in \cite{GLN} investigating the economic effects of infectious diseases. 

More precisely, according to \cite{CPT},  the financial transactions in \cite{DPTZ} were based on the choice of two parameters. The first defines the so-called safeguard threshold, i.e.  the maximum percentage of money that the individual is willing to employ in a transaction, and the second the random risk inherent in the transaction characterized by its variance  through a spread proportional to the square of the individual's wealth. There, the time dependence of the variance was postulated by assuming that, in the presence of a significant epidemic spread, the variance of the risk tends to increase. This is in agreement with the financial market reactions that were often observed during the Covid-19 pandemic to announcements of rising numbers of infected people in several countries \cite{ZHJ}. Thanks to the model in \cite{DPTZ}, it was possible to qualitatively observe the effects of the pandemic in terms of the reduction of the middle class and the increase of social inequalities (see also \cite{Deaton, vB}).

The possibility, starting in early 2021, of launching a widespread vaccination campaign has led to general optimism about the ability to improve economic performance while limiting the health consequences of the epidemic. However, it is clear that the reduction of economic consequences is closely linked to the effectiveness of the vaccine in containing the infections. 

In this paper we will focus, at the level of wealth distribution, on the economic improvements induced by the vaccination campaign in terms of its percentage of effectiveness. The interplay between the economic trend and the pandemic will be evaluated by resorting to a mathematical model joining a kinetic model of wealth distribution based on binary transactions with a compartmental epidemic model including vaccinated individuals (see also \cite{GGHH}). In particular, a fraction of vaccinated individuals, which is determined by the efficacy of the vaccine, may contract the disease. Without intending to review the extensive literature on this topic we cite the recent papers \cite{Dolgin_Nature,Tow,BDOG,Colombo} that highlight the possible partial  immunity provided by current vaccinations. Moreover, the emergence of viral variants makes efficacy of the vaccine inherently non-constant and subject to the collective compliance to non-pharmaceutical interventions. 

The underlying theoretical framework we consider is that of kinetic models for collective social phenomena, which allows for the linking of microscopic agent-based behaviour to emerging observable patterns \cite{PT13}. In particular, mathematical modelling of wealth distribution has seen in recent decades a marked development \cite{BST,ChaCha00,ChChSt05,CPP,DY,DMT,Gup,CC1,Furioli,BM} from which, at least partially, the essential economic mechanisms that are responsible for the formation of large scale economic indicators like Pareto or Gini index have been understood \cite{MS,DPT,PT06}. 

The interplay between epidemic spread and the social economic background is described here as the result of interactions among a large number of individuals, each of which is characterized by the variable $w\in \mathbb R_+$ measuring the amount of wealth of a single agent. In this regard, as shown in \cite{Albi_etal,DPTZ,DPeTZ,DTZ,Zanella}, the fundamental tools of statistical physics allow the understanding of epidemiological dynamics by linking classical compartmental approaches with a statistical description of economic aspects. Indeed, the  multiscale nature of kinetic theory allows for the determination of the macroscopic (or aggregate) and measurable features of disease evolution \cite{Bellomo, LT21, PT13}. 

The rest of the paper is organized as follows. Section \ref{model} introduces the SIR-type system of kinetic equations that includes vaccinated individuals and combines the dynamics of wealth evolution with the spread of infectious disease in a system of interacting agents. Next, in Section \ref{properties}  we study the main mathematical properties of the system, and show that, through a suitable asymptotic procedure, the solution of the kinetic system tends to the solution of a system of Fokker-Planck type equations which exhibits explicit equilibria of inverse Gamma-type.  Finally, in Section \ref{numerics}, we investigate numerically the solutions of the Boltzmann-type kinetic system, and its Fokker-Planck asymptotics, along with the evolution of the Gini index,  characterizing the wealth inequalities. These simulations confirm the model's ability to describe phenomena characteristic of economic trends in situations compromised by the rapid spread of an epidemic, and their variations as a function of the effectiveness of the vaccination campaign.

\section{Wealth dynamics in epidemic phenomena}\label{model}

In this Section we present an extension of the  SIR-kinetic  compartmental description of epidemic spreading introduced in \cite{DPTZ}, which additionally takes into account the population of vaccinated individuals. The model consists of a system of four kinetic equations describing the evolution of wealth in presence of an infectious disease with partial efficacy of vaccination. The entire population is divided then into four compartments: susceptible individuals ($S$), who can  contract the disease; identified infectious individuals ($I$), who are recognized to have contracted the disease and can transmit it; vaccinated individuals ($V$), who have received a vaccine, but can still be, at least partially, infected and contagious, and the recovered individuals ($R$), who are healed and immune. The model can be easily adapted to include disease-related mortality and other compartments of interest in terms of available data, such as hospitalized individuals. We refer to \cite{APZ21,DTZ,H,Par21,Gatto} and the references therein for possible developments in these directions. It should be noted that, since we are referring to an advanced epidemic situation in which we assume the existence of a vaccine, the dynamics of unidentified asymptomatic individuals, so significant in the early stages of the Covid-19 pandemic, has become less relevant thanks to mass screening programs. For this reason we have chosen to employ only one compartment $I$ related to the identified infected individuals. To measure the aggregate effects of vaccination over the whole population we have considered the compartment $V$ with a given vaccine efficacy.  

The agents of each compartment are characterized uniquely by their wealth $w \ge 0$. Hence, we denote by $f_H(w,t)$, $H \in \{S,I,V,R\}$, the distributions of wealth at time $t \ge0$  in each compartment, such that $f_H(w,t)dw$ denote the fraction of agents belonging to the compartment $J$ which, at time $t \ge0$, are characterized by wealth between $w$ and $w + dw$. The total wealth distribution density is then defined by the sum of the distributions in all compartments
\[
f(w,t) = f_S(w,t) + f_I(w,t) + f_V(w,t) + f_R(w,t), \qquad \int_{\mathbb R_+} f(w,t)dw = 1, 
\]
for all $t \ge 0$. Hence, the fractions of the population belonging to each compartment is given by
\[
J(t) = \int_{\mathbb R_+} f_J(w,t)dw, \qquad J \in \{S,I,V,R\}. 
\]

We denote by $m_{J,\kappa}(t)$ the local momenta of order $\kappa$ for the wealth distributions in each compartment
\be\label{mk}
m_{\kappa,J}(t) = \dfrac{1}{J(t)} \int_{\mathbb R_+}w^\kappa f_J(w,t)dw,
\ee
and we denote with $m_\kappa(t)$ the moment of order $\kappa> 0$ of the wealth distribution $f(w,t)$
\[
m_{\kappa}(t) = \int_{\mathbb R_+}w^\kappa f(w,t)dw = \sum_{J \in \{S,I,V,R\}} J(t)\,m_{\kappa,J}(t).
\]

\subsection{The kinetic model}
Following \cite{DPTZ}, we assume that the evolution of the densities obeys a SIR-type compartmental model and that the wealth exchange process is influenced by the epidemic dynamics. This gives a system of four kinetic equations for the unknown distributions $f_H(w,t)$, $H \in \{S,I,V,R\}$, expressed by
\be
\label{eq:kin_SIVR}
\begin{split}
\partial_t f_S(w,t) &= -K(f_S,f_I)(w,t) -\alpha f_S(w,t)+ \sum_{J \in\{S,I,V,R\}} Q_{SJ}(f_S,f_J)(w,t),  \\
\partial_t f_I(w,t) &= K(f_S,f_I)(w,t) +(1-\zeta) K(f_V,f_I)(w,t) - \gamma_I f_I(w,t) + \sum_{J \in\{S,I,V,R\}} Q_{IJ}(f_I,f_J)(w,t) , \\
\partial_t f_V(w,t) &= \alpha f_S(w,t) - (1-\zeta) K( f_V,f_I)(w,t) + \sum_{J \in\{S,I,V,R\}} Q_{VJ}(f_V,f_J)(w,t), \\
\partial_t f_R (w,t) &= \gamma_I f_I(w,t) + \sum_{J \in\{S,I,V,R\}} Q_{RJ}(f_R,f_J)(w,t), 
\end{split}\ee
where $\gamma\ge0$ is the recovery rate for the infected  compartment, $\alpha \in [0,1]$ is the vaccination rate of individuals, while the term $0\le 1-\zeta \le 1$ quantifies the effectiveness of the vaccine, in such a way that high effectiveness corresponds to values close to one of the parameter $\zeta$. The operator $K(\cdot,\cdot)$ governs the transmission of the infection and is considered of the following form 
\be\label{eq:K}
K(f_H,f_I)(w,t) = f_H(w,t) \int_{\mathbb R_+} \beta(w,w_*) f_I(w_*,t)\,dw_*,
\ee
for any $H \in \{S,I,V,R\}$. In \eqref{eq:K} the function $\beta(w,w_*)\ge 0$ denotes the contact rate between people with wealth $w$ and, respectively,  $w_*$. A leading example for $\beta(w,w_*)$ is obtained by choosing analogously to \cite{DPTZ}
\begin{equation}
\label{eq:beta}
\beta(w,w_*) = \dfrac{\bar\beta}{(c + |w-w_*|)^\nu},
\end{equation}
where $\bar \beta>0$, $\nu>0$ and $c\ge 0$. According to the above contact rate, agents with similar wealth are more likely to interact. 

Finally, the operators $Q_{HJ}(f_H,f_J)$, $H,J\in\{S,I,V,R\}$,  characterize the evolution of the wealth evolution in each compartment due to wealth exchange activities between agents of the same class, or between agents of different classes $H$ and $J$. Their form follows the one originally proposed in the Cordier-Pareschi-Toscani model \cite{CPT}. An interaction between two individuals in compartment $H$ and $J$ with wealth pair $(w,w_*)$ leads to a wealth pair $(w^\prime_{JH},w^\prime_{HJ})$ defined by relations
\be\label{eq:binary}
\begin{split}
w^\prime_{HJ} &= (1-\lambda_H)w + \lambda_J w_*+ \eta_{HJ}w \\
w^\prime_{JH} &= (1-\lambda_J)w_* + \lambda_H w+ \eta_{JH}w_*,
\end{split}
\ee
with $H,J \in \{S,I,V,R\}$. In \eqref{eq:binary}  the constants $\lambda_H,\lambda_J \in (0,1)$ are exchange parameters defining the saving propensities $1-\lambda_H$ and $1-\lambda_J$, i.e.  the maximum percentage of money that individuals are willing to employ in a general monetary transaction. Note that the parameters are different in each compartment,  underlining the different behavior of agents  in presence of the pandemic. The choice $\lambda_V>\lambda_S$ for example reflects the fact that susceptible non vaccinated agents have reduced action in wealth exchanges due to various government restrictions with respect to vaccinated individuals.

Furthermore, $\eta_{JH} \ge -\lambda_H, \eta_{HJ} \ge -\lambda_J$ are  independent centered random variables with the same distribution $\Theta$ such that $\textrm{Var}(\eta_{HJ}) =\textrm{Var}( \eta_{HJ})= \sigma^2(t)$. The quantity $\sigma^2(t)$ represents the market risk, which is the same for the whole population and is influenced by the progress of the pandemic. This is in agreement to market reactions that have been observed during new epidemic waves, see e.g. \cite{ZHJ}. It is convenient to express the operators $Q_{HJ}(f_H,fJ)$ in weak form, i.e. the way these operators act on observable quantities \cite{PT13}. 

Let $\varphi(w)$ be a test function and let $\left\langle \cdot \right\rangle$ denote the expectation with respect to the pair of random variables $\eta_{JH}, \eta_{HJ}$ in the interaction process \eqref{eq:binary}. Then, for $H,J \in \{S,I,V,R\}$ we define the Boltzmann-type bilinear operators as follows
\be
\label{eq:Q}
\begin{split}
\int_{\mathbb R_+}\varphi(w)Q_{HJ}(f_H,f_J)(w,t) dw
 = \left\langle \int_{\mathbb R_+^2}(\varphi(w^\prime_{HJ})-\varphi(w)) f_H(w,t)f_J(w_*,t)dw\,dw_* \right\rangle
\end{split}
\ee
being $(w, w_*) \rightarrow (w_{JH}^\prime ,w_{HJ}^\prime)$ as in \eqref{eq:binary} and where  $\left \langle \cdot \right\rangle$ denotes the expectation with respect to the independent random variables $\eta_{HJ}, \eta_{HJ}$. 

Binary interactions between individuals \eqref{eq:binary} reflect the idea that wealth exchanges occur between pairs of agents who invest a fraction of their wealth in the presence of an equivalent good. In each case, such investments involve nondeterministic speculative risks that can provide additional wealth or loss of wealth. The aggregate behavior of the population is then provided by the operators \eqref{eq:Q} from which we obtain the emerging macroscopic trends of the binary exchanges considered in each epidemiological compartment. 

\begin{remark}
In the kinetic epidemic model \eqref{eq:kin_SIVR} the passage from susceptible to vaccinated is governed by a very simple dynamics that does not take into account possible vaccine limitations, as in the first phase of the vaccination campaign. In general, the vaccination rate $\alpha$ may depend on several factors like age, work status of individuals and time. 
It is worth to observe that, in addition to the natural dependency of the recovery rate $\gamma_I$ from age \cite{APZ21, APZ2, Colombo}, we may also consider wealth-dependent recovery rates to take into account that high wealth can provide access to better hospitals in some health systems, ensuring thus higher chance of recovery \cite{DPTZ}. We point the interested reader to \cite{Zanella} for a more detailed discussion based on available data.
\end{remark}

\subsection{Evolution of macroscopic quantities}\label{macro}
In the following we discuss the evolution of emerging macroscopic quantities form the kinetic model \eqref{eq:kin_SIVR}. Let $\varphi(w)$ be a test function. Choosing $\varphi(w) = 1$ in \eqref{eq:Q} we have
\[
\sum_{J \in \{S,I,V,R\}}\int_{\mathbb R_+}\varphi(w) Q_{HJ}(f_H,f_J)(w,t)dw = 0,
\]
which correspond to mass conservation, i.e. the conservation of the number of agents. If $\varphi(w)=w$ in \eqref{eq:Q} we get the evolution of the average wealth in each compartment, corresponding to the first quantity not conserved in time
\be
\label{eq:w1}
\begin{split}
\dfrac{d}{dt}m_{1,H}(t) &= \dfrac{1}{H(t)} \sum_{J\in\{S,I,V,R\}} \int_{\R_+^{2}}\langle w^\prime_{HJ}-w \rangle f_H(w,t)f_J(w_{*},t)dw dw_{*} \\
&=H(t) \sum_{J\in\{S,I,V,R\}}J(t)(\lambda_J m_{1,J}(t)-\lambda_Hm_{1,H}).
\end{split}
\ee
The total mean wealth is then conserved 
$$\dfrac{d}{dt}\sum_{H \in \{S,I,V,R\}} \int_{\mathbb R_+} w f_H(w,t)dw = \dfrac{d}{dt}m_{1} = 0.$$
The evolution of mass fractions can be easily obtained from \eqref{eq:kin_SIVR} by direct integration 
\begin{equation}
\label{eq:mass}
\begin{split}
\dfrac{d}{dt} S(t) &=- \int_{\mathbb R_+^2} \beta(w,w_*)f_S(w,t)f_I(w,t)dw\,dw_* - \alpha S(t),  \\
\dfrac{d}{dt} I(t) &= \int_{\mathbb R_+^2} \beta(w,w_*)f_S(w,t)f_I(w,t)dw\,dw_*  + (1-\zeta) \int_{\mathbb R_+^2} \beta(w,w_*)f_V(w,t)f_I(w,t)dw\,dw_*  - \gamma_I I(t), \\
\dfrac{d}{dt} V(t) &= \alpha S(t) - (1-\zeta)\int_{\mathbb R_+^2} \beta(w,w_*)f_V(w,t)f_I(w,t)dw\,dw_*, \\
\dfrac{d}{dt} R(t) &=\gamma_I I(t). 
\end{split}
\end{equation}
To get a closed form evolution of the macroscopic quantities we consider a constant rate function, $\beta(w,w_{*})=\bar\beta>0$, obtained from \eqref{eq:beta} for $\nu = 0$, and a constant in time market risk $\sigma^2(t)=\sigma^2$. Under these assumptions, thanks to mass conservation of Boltzmann-type operators \eqref{eq:Q} we get a classical SIR model with vaccination
\begin{equation}
\label{eq: SIVR}
\begin{split}
\dfrac{d}{dt} S(t) &=- \bar\beta S(t)I(t) - \alpha S(t),  \\
\dfrac{d}{dt} I(t) &= \bar \beta S(t)I(t)  + (1-\zeta) \bar\beta V(t)I(t)  - \gamma_I I(t), \\
\dfrac{d}{dt} V(t) &= \alpha S(t) - (1-\zeta)\bar\beta V(t)I(t), \\
\dfrac{d}{dt} R(t) &=\gamma_I I(t). 
\end{split}
\end{equation}
As a consequence for large times $t\to +\infty$ we have a disease free equilibrium state where $I(t)\to 0^+$, $S(t) \to 0^+$, $V(t)\to V^\infty$ and $R(t) \to R^\infty$ with $V^\infty + R^\infty=1$ (see \cite{H}).

The dynamics of mean wealths can be recovered from \eqref{eq:w1} as follows
\be
\label{eq:mean}
\begin{split}
S(t)\frac{d }{dt}m_{1,S}(t) &= S(t)( \bar m_1(t) -\lambda_S m_{1,S}(t)), \\
I(t)\frac{d}{dt}m_{1,I}(t) &=  \bar \beta S(t)I(t) (m_{1,S}-m_{1,I}) + \bar\beta (1-\xi)V(t) I(t) (m_{1,V} - m_{1,I})\\
&\quad + I(t) (\bar m_1 - \lambda_I m_{1,I}), \\
V(t)\frac{d}{dt}m_{1,V}(t) &= \alpha S(t) (m_{1,S}-m_{1,V}) + V(t) (\bar m_1 - \lambda_V m_{1,V}), \\
R(t)\frac{d}{dt}m_{1,R} (t) &= \gamma_I I(t)(m_{1,R}(t)-m_{1,I}(t)) +R(t)(\bar m_1(t)-\lambda_R m_{1,R}(t)), 
\end{split}
\ee
where we defined the weighted mean wealth
\be
\label{eq:totmean}
\bar m_1(t)=\sum_{J\in\{S,I,V,R\}}\lambda_Jm_{1,J}(t)J(t).
\ee
Therefore, from \eqref{eq:mean} we obtain that the large time behavior of the mean wealth satisfies 
\[
2\bar m_1^\infty - \lambda_V m_{1,V}^\infty  - \lambda_R m_{1,R}^\infty = 0. 
\]
Hence, we obtain
\[
\lambda_Vm^{\infty}_{1,V}=\lambda_Rm^{\infty}_{1,R},
\]
together with the constraint $R^{\infty}m^{\infty}_{R,1}+V^{\infty}m^{\infty}_{V,1}=m$ from the conservation of total mean wealth. Thanks to the last equalities we have that the asymptotic mean wealth in the compartments of vaccinated and recovered individuals are given by
\be
\label{eq:asymean}
    m^{\infty}_{1,V}=\frac{\lambda_R}{\lambda_RV^{\infty}+\lambda_VR^{\infty}}m, \quad\quad
    m^{\infty}_{1,R}=\frac{\lambda_V}{\lambda_RV^{\infty}+\lambda_VR^{\infty}}m.
\ee
Likewise, we obtain the system for the the second moments
\be
\label{eq:energy}
\begin{split}
S(t)\dfrac{d}{dt}m_{2,S}(t) &= (\lambda_S^2-2\lambda_S + \sigma^2)Sm_{2,S} + S(t)\bar m_2+ 2(1-\lambda_S)Sm_{1,S}\bar m_1, \\
I(t) \dfrac{d}{dt} m_{2,I}(t) &= \bar \beta SI(m_{2,S}-m_{2,I}) + (1-\zeta)\bar\beta VI(m_{2,V} - m_{2,I}) \\
&\quad+ (\lambda_I^2 -2\lambda_I + \sigma^2) Im_{2,I} + I\bar m_2 + 2(1-\lambda_I) I m_{1,I}\bar m_1,\\
V(t) \dfrac{d}{dt} m_{2,V}(t) &= \alpha S (m_{2,S} - m_{2,V}) + (\lambda_V^2 - 2\lambda_V + \sigma^2) V m_{2,V} + V \bar m_2\\
&\quad + 2(1-\lambda_V)Vm_{1,V}\bar m_1, \\
R(t) \dfrac{d}{dt} m_{2,R}(t)& = (\lambda_R^2  - 2\lambda_R + \sigma^2) R m_{2,R} + R \bar m_2 + 2(1-\lambda_R)Rm_{1,R}\bar m_1,
\end{split}
\ee
where $\bar m_1$ has been defined in \eqref{eq:totmean} and we have introduced the following notation
\[
\bar m_2(t)=\sum_{J\in\{S,I,V,R\}}\lambda_J^2m_{2,J}(t)J(t).
\]
 The evolution of the second moment for the whole system is governed by
\[
\frac{d}{dt}m_{2}(t)=\bar m_2(t)+\sum_{J\in\{S,I,V,R\}}\left(m_{J,2}(\lambda_J^2-2\lambda_J+\sigma)+2(1-\lambda_J)m_{J,2}\bar m_1(t)\right)J(t).
\]
For large times the second order moment for susceptible and infected is such that $m_{2,S},m_{2,I} \rightarrow 0^+$ for $t \rightarrow +\infty$. Therefore, $m_{2,V}^\infty$, $m^\infty_{2,R}$ are solutions to
\[
\begin{split}
(\lambda_V^2 - 2\lambda_V +\sigma^2) m_{2,V}^\infty + \bar m_2^\infty + (1-\lambda_V)m_{1,V}^\infty \bar m_1^\infty = 0,\\
(\lambda_R^2 - 2\lambda_R +\sigma^2) m_{2,R}^\infty + \bar m_2^\infty + (1-\lambda_R)m_{1,R}^\infty \bar m_1^\infty = 0.
\end{split}\]
from which we get
\[
\begin{split}
m_{2,R}^\infty = \dfrac{\lambda_V^2 (1-\lambda_V)V^\infty m_{1,V}^\infty \bar m_1^\infty  - A_V(1-\lambda_R)m_{1,R}^\infty \bar m_1^\infty}{A_V(\lambda_R^2(1+R^\infty) - 2\lambda_R + \sigma^2 ) - \lambda_V^2 \lambda_R^2 V^\infty R^\infty} \\
m_{2,V}^\infty = \dfrac{\lambda_R^2 (1-\lambda_R)R^\infty m_{1,R}^\infty \bar m_1^\infty  - A_R(1-\lambda_V)m_{1,V}^\infty \bar m_1^\infty}{A_R(\lambda_V^2(1+V^\infty) - 2\lambda_V + \sigma^2 ) - \lambda_R^2 \lambda_V^2 V^\infty R^\infty}
\end{split}
\]
where 
\[
A_H= \lambda_V^2 (1+H^\infty) - 2\lambda_V + \sigma^2, \qquad H \in \{V,R\},
\]
and $\bar m_1^\infty = \lambda_V m_{1,V}^\infty V^\infty + \lambda_R m_{1,R}^\infty R^\infty $ and $m_{1,V}^\infty, m_{1,R}^\infty$ have been obtained in \eqref{eq:asymean}. 

\begin{remark}
In the general case where a non constant incidence rate $\beta=\beta(w,w_*)$ is considered the macroscopic system of equations is not closed. Depending on the specific choice of $\beta$ and using the knowledge of the equilibrium states discussed in Section \ref{subsection 3.1} it is possible, through the classical hydrodynamic closure of kinetic theory, to derive epidemic models where the dynamics, instead of being homogeneous as in classical compartmental modeling, is influenced by the heterogeneous wealth status of individuals. We refer to \cite{Albi_etal, DPeTZ} for examples in this direction.
\end{remark}

\section{Properties of the kinetic model} \label{properties}
In this section we study the mathematical model \eqref{eq:kin_SIVR} from an analytical point of view, by proving well-posedness and convergence to equilibrium of the solution. To this aim we made suitable simplification assumptions on the contact rate by restricting to the case  $\beta(w,w_*) = \bar \beta$. We resort to classical mathematical approaches for kinetic equations to characterize the trend to equilibrium \cite{DPTZ,PT13}. In particular, taking into account methods for nonconservative systems, see e.g. \cite{BCT}, we provide an existence and uniqueness result. Given a function $f(w)\in L^1(\R_+)$, we define its Fourier transform as follows
\[
\hat{f}(z)=\int_{\R}e^{-iwz}f(w)dw.
\]
Under the above assumption on the contact rate, we rewrite \eqref{eq:kin_SIVR} in weak form
\be
\label{eq:kinweak_SIVR}
\begin{split}
\partial_t \int_{\R_+}\varphi(w) f_S(w,t)dw &= -\bar\beta I(t)\int_{\R_+}\varphi(w)f_S(w,t)dw -\alpha \int_{\R_+} \varphi(w)f_S(w,t)dw \\
&+ \sum_{J \in\{S,I,V,R\}} \int_{\R_+}\varphi(w)Q_{SJ}(f_S,f_J)(w,t)dw,  \\
\partial_t \int_{\R_+}\varphi(w) f_I(w,t)dw &= \bar\beta I(t)\int_{\R_+}\varphi(w)f_S(w,t)dw +(1-\zeta) \bar\beta I(t) \int_{\R_+}\varphi(w)f_V(w,t)dw\\
& - \gamma_I \int_{\R_+}\varphi(w)f_I(w,t)dw + \sum_{J \in\{S,I,V,R\}}\int_{\R_+}\varphi(w) Q_{IJ}(f_I,f_J)(w,t) dw, \\
\partial_t \int_{\R_+}\varphi(w) f_V(w,t)dw &= \alpha \int_{\R_+}\varphi(w) f_S(w,t)dw - (1-\zeta) \bar\beta I(t) \int_{\R_+}\varphi(w)f_V(w,t)dw \\
&+ \sum_{J \in\{S,I,V,R\}} \int_{\R_+}\varphi(w)Q_{VJ}(f_V,f_J)(w,t)dw, \\
\partial_t \int_{\R_+}\varphi(w) f_R (w,t)dw &= \gamma_I \int_{\R_+}\varphi(w) f_I(w,t)dw + \sum_{J \in\{S,I,V,R\}} \int_{\R_+}\varphi(w)Q_{RJ}(f_R,f_J)(w,t)dw. 
\end{split}\ee
Hence, we consider $\varphi(w) = e^{-izw}$ in \eqref{eq:kinweak_SIVR} to get
\be
\label{eq: FoSIAR}
\begin{split}
 \partial_t \hat{f}_{S}(z,t) &= -\bar\beta I(t)\hat{f}_{S}(z,t)-\alpha \hat{f}_S(z,t) + \sum_{J\in\{S,I,V,R\}}\hat{Q}(\hat{f}_S,\hat{f}_J)(z,t),\\
 \partial_t \hat{f}_{I}(z,t) &=\bar\beta I(t)\hat{f}_{S}(z,t) +(1-\zeta)\bar\beta\hat{f}_I(z,t)\hat{f}_V(z,t) -\gamma_{I}\hat{f}_{I}(z,t) + \sum_{J\in\{S,I,V,R\}}\hat{Q}(\hat{f}_I,\hat{f}_J)(z,t),\\
 \partial_t \hat{f}_{V}(z,t) &= \alpha\hat{f}_{S}(z,t)-(1-\zeta)\bar\beta\hat{f}_{I}(z,t)\hat{f}_{V}(z,t)+ \sum_{J\in\{S,I,V,R\}}\hat{Q}(\hat{f}_V,\hat{f}_J)(z,t),\\
 \partial_t \hat{f}_{R}(z,t) &= \gamma_{I}\hat{f}_{I}(z,t) +  \sum_{J\in\{S,I,V,R\}}\hat{Q}(\hat{f}_R,\hat{f}_J)(z,t).\\
 \end{split}
\ee
Similarly to \cite{DPTZ} the operators $\hat{Q}(\hat{f}_{H},\hat{f}_{J})(z,t)$ may be rewritten as follows
\[
\int_{\R_+}e^{-iwz}Q(f_H,f_J)dw=\langle\hat{f}_H(A_{HJ}z,t)\rangle\hat{f}_J(\lambda_{J}z,t)-J(t)\hat{f}_H(z,t),
\]
where
\[
A_{HJ}=1-\lambda_H+\eta_{HJ}.
\]
We assume that the parameters of the trading activity satisfy the condition
\be
\label{eq: nu}
\nu=\max_{H,J\in\{S,I,V,R\}}[\lambda_J^2+\langle A_{HJ}^2\rangle]<1.
\ee
Let $\mathcal{P}_s(\R_+)$ be the set of probability measures $f(w)$ with bounded $s-$moment,  and, for any pair of densities $f$ and $g$ in   $\mathcal{P}_s(\R_+)$ let us consider the class of metrics $d_s$ defined by
\be
\label{eq: Fdistance}
d_s(f,g)=\sup_{z\in\R}\frac{|\hat{f}(z)-\hat{g}(z)|}{|z|^s},
\ee
where $\hat f$ and $\hat g$ denote the Fourier transforms of $f$ and $g$. Then the distance \eqref{eq: Fdistance} is well-defined and finite for any pair of probability measures with equal moments up to order $[s]$\footnote{ where $[s]$ denotes the integer part of $s$}, if $s$ is a real number or up to $s-1$, if $s$ is an integer \cite{PT13,TV99}. 

Inequality \eqref{eq: nu} combined with a Fourier-based distance allows  to obtain an exponential convergence to  equilibrium for system \eqref{eq:kin_SIVR}. This condition is verified whenever
\[
\sigma^2<2\min_{J\in\{S,I,V,R\}}\lambda_J(1-\lambda_J),
\]
namely when the market risk is not too big in relationship withe the saving propensities. To study the large-time behavior of the solution to systems like \eqref{eq: FoSIAR} we follow \cite{DPTZ, PT13}. 

\noindent Then we have the following result 
\begin{theorem}
\label{Theo: Conv}
Let $f_J(w,t)$ and $g_J(w,t)$, $J\in\{S,I,V,R\}$, be two solutions of the kynetic system \eqref{eq:kin_SIVR}, corresponding to the initial values $f_J(w,0)$ and $g_J(w,0)$ such that $d_2(f_J(w,0),g_J(w,0))$, $J\in\{S,I,V,R\}$, is finite. Then, if condition \eqref{eq: nu} holds, the Fourier-based distance $d_2(f_J(w,t),g_J(w,t))$ decays exponentially in time toward zero and the following holds
\be
\label{eq: conv}
\sum_{J\in\{S,I,V,R\}}d_2(f_J(w,t),g_J(w,t)) < \sum_{J\in\{S,I,V,R\}}d_2(f_J(w,0),g_J(w,0))\exp\{-(1-\nu)t\}.
\ee
\end{theorem}
\noindent The previous result and the equation \eqref{eq: conv} give us the contractivity of the system in the $d_2$ metric which will be the essential to prove the existence theorem. Theorem \ref{Theo: Conv} allows us to further investigate the properties of the steady state $f_J^{\infty}(w)$, $J\in\{S,I,V,R\}$.
 

\noindent In order to obtain an existence result we need to introduce a subset of $\mathcal{P}_2(\R)$
\be
\label{eq:set}
\mathcal{D}_{m_1,m_2}:=\Bigg\{ F\in\mathcal{P}_2(\R):\;\int_{\R}vdF(v)=m_1,\;\int_{\R}v^2dF(v)=m_{2}\Bigg\}.
\ee

Following \cite{TV99} it is possible to prove that $\mathcal{D}_{m_1,m_2}$ is a metric Banach space with the $d_2(\cdot,\cdot)$ metric. Now, we define
\[
\mathcal{D}^{\infty}:=\mathcal{D}_{m^{\infty}_{V,1},m^{\infty}_{V,2}}\times\mathcal{D}_{m^{\infty}_{R,1},m^{\infty}_{R,2}}
\]
as the product space of two sets like \eqref{eq:set} where the momenta are those of the steady states for the relative distributions $f_J(w)$, for $J\in\{V,R\}$ (we are only considering these two classes since for large time $I,S\rightarrow0^+$). We also recall a variant of the metric used in Theorem \ref{Theo: Conv} 
\be
\label{distance}
\overline{d}_2(f,g):=\sum_{J\in\{V,R\}}d_2(f_J(w,t),g_J(w,t)).
\ee
Now, we are able to prove 
\begin{theorem}
\label{Theo: Exist}
If the initial value $f_0(w)=f(w,0)\in\mathcal{D}^{\infty}$ and condition \eqref{eq: nu} holds then, the system
\be
\label{eq:kin_SIVR1}
\begin{split}
	\partial_t f_V(w,t)  & = \sum_{J \in\{V,R\}} Q_{VJ}(f_V,f_J)(w,t), \\
	\partial_t f_R (w,t) &=  \sum_{J \in\{V,R\}} Q_{RJ}(f_R,f_J)(w,t), 
\end{split}\ee
 has a unique steady state $f^{\infty}(w)$, and it also belongs to $\mathcal{D}^{\infty}$.
\end{theorem}
\begin{proof}
Let us consider the flow map
\be
\label{eq:flow}
T_t:\Big(\mathcal{D}^{\infty},\bar d_2\Big)\to\Big(\mathcal{D}^{\infty},\bar d_2\Big)
\ee
which , for any time $t>0$, is given by $T_t(f_0(w))=f(t)=(f_V(w,t), f_R(w,t))$ where $f(t)$ is the solution of \eqref{eq:kin_SIVR1} at time $t$ with $f(w,0)=f_0(w)\in\mathcal{D}^{\infty}$. Thanks to \eqref{eq: conv} we have
\[
\overline{d}_2\big(T_t(f_0(w)),\;T_t(g_0(w))\big) < \overline{d}_2\big(f_0(w),\;g_0(w)\big)\exp\{-(1-\nu)t\}
\]
which is a strict contraction for \eqref{eq:flow} with constant $\exp\{-(1-\nu)t\}<1$. Now, it is easy to see that $\big(\mathcal{D}^{\infty}, \overline{d}_2\big)$ is a Banach space and therefore, Banach fixed point theorem ensures the existence and uniqueness for the steady state in $\mathcal{D}^{\infty}$.
\end{proof}

\begin{remark}
Similar results may be obtained in the more realistic case $\beta(w,w_*) = \beta(w-w_*)$ since the transmission operator $K(\cdot,\cdot)$ defined in \eqref{eq:K} possesses, in this case, a convolution structure which naturally converts into a product in the Fourier space. We omit the details. 
\end{remark}

\subsection{Fokker-Planck scaling and steady states}
\label{subsection 3.1}

In the general case, it is difficult to compute analytically the large time behaviour of the compartmental kinetic system \eqref{eq:kin_SIVR}. 
A deeper insight on the steady states can be obtained through the so-called quasi-invariant limit procedure \cite{CPT, DPTZ, FPTT16, PT13}. The goal is to derive a simplified Fokker-Planck model for which the study of the asymptotic properties is much easier. It is worth to mention that this approach  is inspired by the so-called grazing collision limit of the Boltzmann equation, see \cite{Cer,Villani}. 

The driving idea is to scale at the same time interactions and trading frequency. As a consequence, the equilibrium of the wealth distribution is reached  faster with respect to the time scale of the epidemic. Hence, given $\epsilon\ll1 $ we introduce the following scaling 
\be
\label{eq:scaling}
\begin{split}
\lambda_{S}\rightarrow\epsilon\lambda_{S}, 
\quad\lambda_{I}\rightarrow\epsilon\lambda_{I},\quad\lambda_{V}\rightarrow\epsilon\lambda_{V},
\quad\lambda_{R}\rightarrow\epsilon\lambda_{R},\\
\sigma^2\rightarrow\epsilon\sigma^2,\quad\beta(w,w_*)\rightarrow\epsilon\beta(w,w_*),\quad\gamma_{I}\rightarrow\epsilon\gamma_{I},
\end{split}
\ee
together with the time scaling $t \rightarrow t/\epsilon$.  We denote as $Q^{\epsilon}_{HJ}(\cdot,\cdot)$, $H,J\in\{S,I,V,R\}$, the scaled interaction terms. Using a Taylor expansion for small values of $\epsilon$ we get \cite{DPTZ} 
\[
\begin{split}
&\dfrac{1}{\epsilon} \int_{\R_+}Q^\epsilon_{HJ}(f_H,f_J)(w,t)\varphi(w)dw \\
&\quad= 
\int_{\R_+}\left\{ -\varphi^\prime(w)(w \lambda_H J - m_{1,J}\lambda_J) + \dfrac{\sigma^2}{2}\varphi^{\prime\prime}(w)w^2 J(t)\right\}f_H(w,t)dw + {O(\epsilon)}. 
\end{split}\]
Integrating back by parts, in the limit $\epsilon\rightarrow 0$ we obtain the system of Fokker-Planck equations
\be
\label{eq:Fokker-Planck}
\begin{split}
\frac{\partial f_S(w,t)}{\partial t}&=-K(f_S,f_I)(w,t) -\alpha f_S(w,t) + \frac{\partial }{\partial w}\{[w\lambda_S-\bar{m}(t)]f_S(w,t)\}\\
&\quad +\frac{\sigma^2}{2}\frac{\partial ^2}{\partial w^2}(w^2f_S(w,t)),\\
\frac{\partial f_I(w,t)}{\partial t}&=K(f_S,f_I)(w,t)+(1-\zeta) K(f_V,f_I)(w,t) - \gamma_{I}f_I(w,t)\\
&\quad + \frac{\partial }{\partial w}\{[w\lambda_I - \bar{m}(t)]f_I(w,t)\} +\frac{\sigma^2}{2}\frac{\partial ^2}{\partial w^2}(w^2f_I(w,t)),\\
\frac{\partial f_V(w,t)}{\partial t}& =\alpha f_S(w,t)-(1-\zeta) K(f_V,f_I)(w,t) +\frac{\partial }{\partial w}\{[w\lambda_V-\bar{m}(t)]f_V(w,t)\}\\ 
&\quad +\frac{\sigma^2}{2}\frac{\partial ^2}{\partial w^2}(w^2f_V(w,t)),\\
\frac{\partial f_R(w,t)}{\partial t}&=\gamma_I(w,t)+\frac{\partial }{\partial w}\{[w\lambda_R-\bar{m}(t)]f_R(w,t)\}+\frac{\sigma^2}{2}\frac{\partial ^2}{\partial w^2}(w^2f_R(w,t)),
\end{split}
\ee
where $\bar m$ has been defined in \eqref{eq:totmean}. The above Fokker-Planck system is complemented with the following boundary conditions
\[
\label{eq:boundary}
\begin{split}
\frac{\partial }{\partial w}[w^2g_J(w,t)]|_{w=0}=0\qquad
[w\lambda_J-\overline{m}]g_J+\frac{\sigma}{2}\frac{\partial }{\partial w}(w^2g_J)\bigg|_{w=0}=0.
\end{split}
\]
We can verify under suitable assumptions that the Fokker-Planck  \eqref{eq:Fokker-Planck} possesses an explicitly computable steady state. Let us consider the case of constant contact rate, i.e. $\beta(w,w_{*})=\bar\beta$. Since for large times $S,I \rightarrow 0^+$ we get that the stationary states $f_V^{\infty}(w)$ and $f_{R}^{\infty}(w)$ solve the following equations
\[
\begin{split}
    \lambda_V\frac{\partial}{\partial w}\bigg[(w-m_V^{\infty}
    )f^{\infty}_V(w)\bigg]+\frac{\sigma^2}{2}\frac{\partial^2}{\partial w^2}[w^2f_V^{\infty}(w)]=0,\\
    \lambda_R\frac{\partial}{\partial w}\bigg[(w-m_R^{\infty}
    )f^{\infty}_R(w)\bigg]+\frac{\sigma^2}{2}\frac{\partial^2}{\partial w^2}[w^2f_R^{\infty}(w)]=0.
\end{split}
\]
From the above equalities we obtain that the two steady states are inverse Gamma densities
\begin{equation}
\label{eq:steady}
f_V^{\infty}(w)=V^{\infty}\frac{\kappa^{\mu_V}}{\Gamma(\mu_V)}\frac{e^{-\frac{\kappa}{w}}}{w^{1+\mu_V}} \quad\quad\quad
f_R^{\infty}(w)=R^{\infty}\frac{\kappa^{\mu_R}}{\Gamma(\mu_R)}\frac{e^{-\frac{\kappa}{w}}}{w^{1+\mu_R}}
\end{equation} 
with Pareto indices defined as follows
\[
\mu_V=1+2\frac{\lambda_{V}}{\sigma^2},\quad\quad\mu_R=1+2\frac{\lambda_R}{\sigma^2},
\]
\[
\kappa=(\mu_V-1)m_V^{\infty}=(\mu_R-1)m_R^{\infty}=\frac{2\lambda_{R}\lambda_{V}}{\sigma^2(\lambda_{R}V^{\infty}+\lambda_{V}R^{\infty})}m. 
\label{kappa}
\]
Consequently, the global steady state is a mixture of inverse Gamma distribution 
\be
\label{eq:global}
f^{\infty}(w)=f_V^{\infty}(w)+f_R^{\infty}(w), 
\ee
which may present a bimodal shape with different intensity. The formation of two peaks at the equilibrium is due to the fact that  we have two different maxima corresponding to the points
\be
\overline{w}_V=\frac{\kappa}{\mu_V+1}=\frac{\lambda_{R}\lambda_{V}}{(\lambda_V+\sigma)(\lambda_RV^{\infty}+\lambda_VR^{\infty})}m,
\label{maxV}
\ee
\be
\overline{w}_R=\frac{\kappa}{\mu_R+1}=\frac{\lambda_{R}\lambda_{V}}{(\lambda_R+\sigma)(\lambda_RV^{\infty}+\lambda_VR^{\infty})}m,
\label{maxR}
\ee
for the vaccinated and for the recovered wealth distributions, respectively. In the next section we report in the resulting profiles for different choices of $\lambda_V$, $\lambda_R$, $\sigma$ and $V^{\infty}, R^{\infty}$. 

\begin{remark}
The emergence of multimodal equilibrium wealth distribution has been classically linked to the appearance of new inequalities in highly stressed societies, see e.g. \cite{Ferrero,Gup,Deaton}. In these cases, the economic segregation of part of the society leads to the pauperisation of substantial layers of the middle class. In the present case, the different economic impact played by agents in each compartment is capable to shape wealth distribution towards a bimodal distribution. Indeed, the trading propensities modelling personal response to the economic scenario, can be substantially modified by the progression of the epidemic and the vaccine efficacy.
\end{remark}


\section{Numerical results}\label{numerics}
In this section we study the impact of vaccination on the equilibrium of the kinetic system through several numerical simulations. This allows us to show the model's ability to describe different situations of wealth distribution in the presence of epidemic dynamics. In particular, we will adopt standard Direct Simulation Monte Carlo  methods to simulate the system of kinetic equations \eqref{eq:kin_SIVR}, see \cite{PR,PT13} and the references therein. In all the subsequent tests we will consider $N= 10^5$ agents and the densities are reconstructed through standard histograms.  

In the first test, we verify numerically the convergence of the solution to the kinetic system \eqref{eq:kin_SIVR} to the solution of the Fokker-Planck system \eqref{eq:Fokker-Planck} under the scaling \eqref{eq:scaling}. Then, we study the emergence of wealth inequalities, measured through the Gini index, in relation to the effectiveness of the vaccine. These results are obtained both in case of a constant market risk variance $\sigma^2$ and in the case of a variance that depends on the current epidemic situation. Lastly, we introduce the possibility that also the effectiveness of the vaccine is affected by the number of positive cases. This situation mimics the realistic case of diffusion of viral variants for which an up-to-date vaccine may be not immediately available. 

\subsection{Test 1: Long time behavior and convergence to equilibrium}
\label{test1}

In this test we want to observe the convergence of the numerical solution of the kinetic system \eqref{eq:kin_SIVR} to the one of the Fokker-Planck system \eqref{eq:Fokker-Planck} in the quasi-invariant limit introduced in Section \ref{subsection 3.1}. We consider the simplified case where $\beta(w,w_*) = \bar \beta = 0.2$,  $\gamma_I = 1/12$, $\zeta = 0.9$ for which we obtained the steady distributions in \eqref{eq:steady}. These values are representative of realistic dynamics during the beginning of the Covid-19 pandemic, see e.g. \cite{Albi_etal,APZ21,APZ2,DT21,Par21,Gatto,Zanella}. 

At time $t = 0$ we consider an uniform distribution over the interval $[0,2]$ 
\begin{equation}
\label{eq:f0}
f(w,0)=\frac{1}{2}\chi(w\in[0,2])
\end{equation}
 where $\chi(\cdot)$ is the indicator function. The distributions of the epidemic compartments are
\begin{equation}
\label{eq:fractions}
f_S(w,0)=\rho_S f(w),\quad f_I(w,0)=\rho_I f(w),\quad f_V(w,0)=\rho_V f(w),\quad f_R(w,0)=\rho_R f(w),
\end{equation}
where the mass fractions are $\rho_I=7.5\times10^{-3}$, $\rho_V=0$, $\rho_R=4\times10^{-2}$ and $\rho_S=1-(\rho_I+\rho_V+\rho_R)$. Furthermore, we consider the value $\sigma^2=0.02$ for the market risk. In Figure \ref{fig:1} we show the numerical solution at time $T=300$ of \eqref{eq:kin_SIVR} in the scaling regime \eqref{eq:scaling} with $\epsilon =1, 0.5, 10^{-3}$. 

In particular, provided an epidemic dynamics such that $V^{\infty}=0.51$ and $R^{\infty}=0.49$, we give numerical evidence of the aforementioned convergence in two regimes expressing increasing safeguard thresholds $1-\lambda_J$, $J \in \{S,I,V,R\}$,  for non-vaccinated agents
\begin{itemize}
\item[$i)$] $\lambda_S=0.15$, $\lambda_I=0.10$, $\lambda_V=0.30$, $\lambda_R=0.20$
\item[$ii)$] $\lambda_S=0.10$, $\lambda_I=0.05$, $\lambda_V=0.30$, $\lambda_R=0.15$
\end{itemize}
 where the same values of $V^{\infty}$ and $R^{\infty}$ are unchanged. 

We observe that, if $\epsilon\ll1$, the Fokker-Planck asymptotic distribution is a consistent approximation of the equilibrium distribution of the Boltzmann-type model. In both cases the global distribution is a mixture of inverse Gamma densities and in the right plot of Figure \ref{fig:1} we can clearly observe a bimodal shape for the wealth distribution. To highlight this we have drawn the maximum points of the distributions $f^{\infty}_V$, $f^{\infty}_R$ which are at $\overline{w}_V$, $\overline{w}_R$ defined in \eqref{maxV}, \eqref{maxR}.

\begin{figure}
    \centering
    \includegraphics[scale = 0.55]{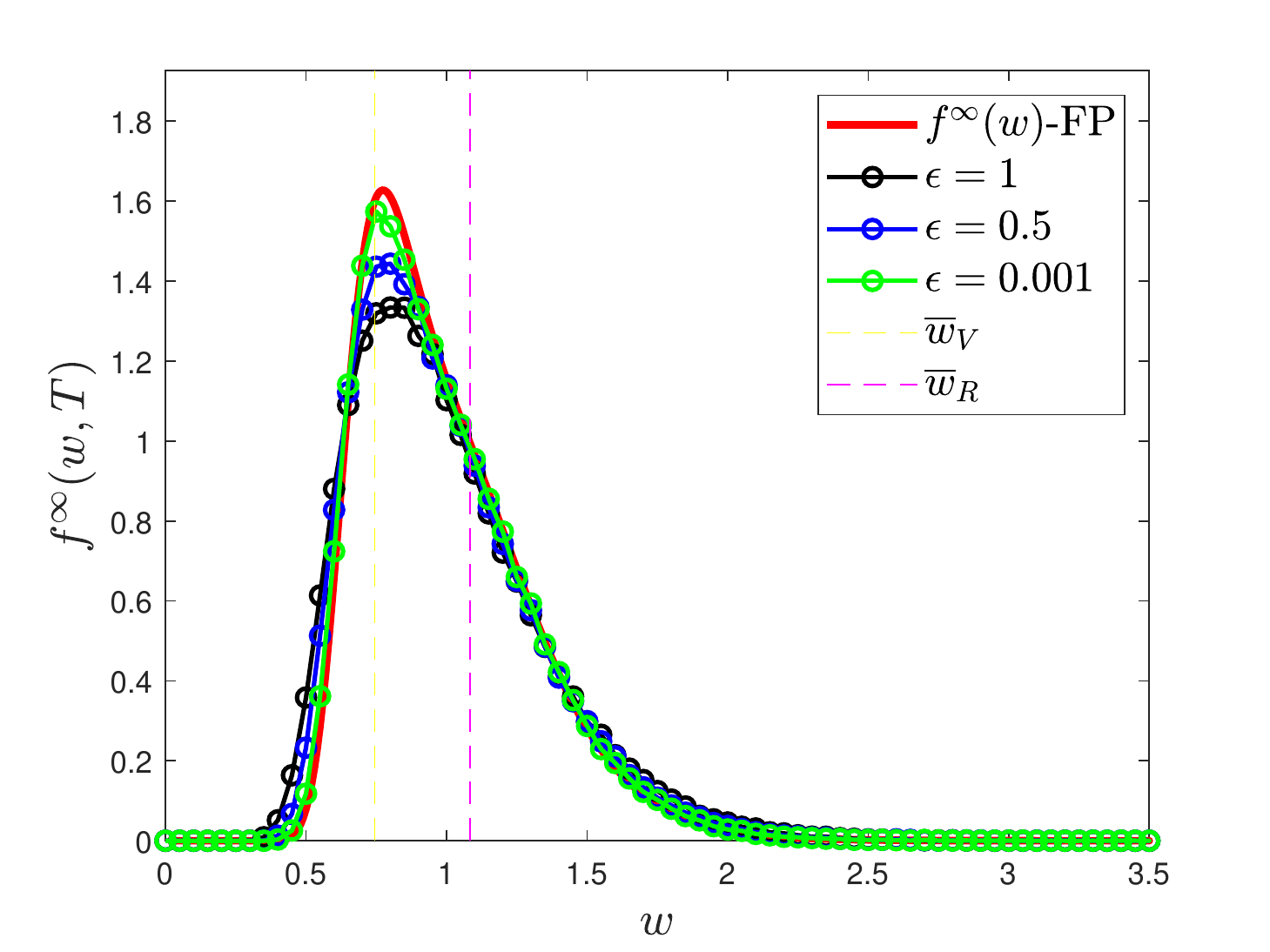}
    \includegraphics[scale = 0.55]{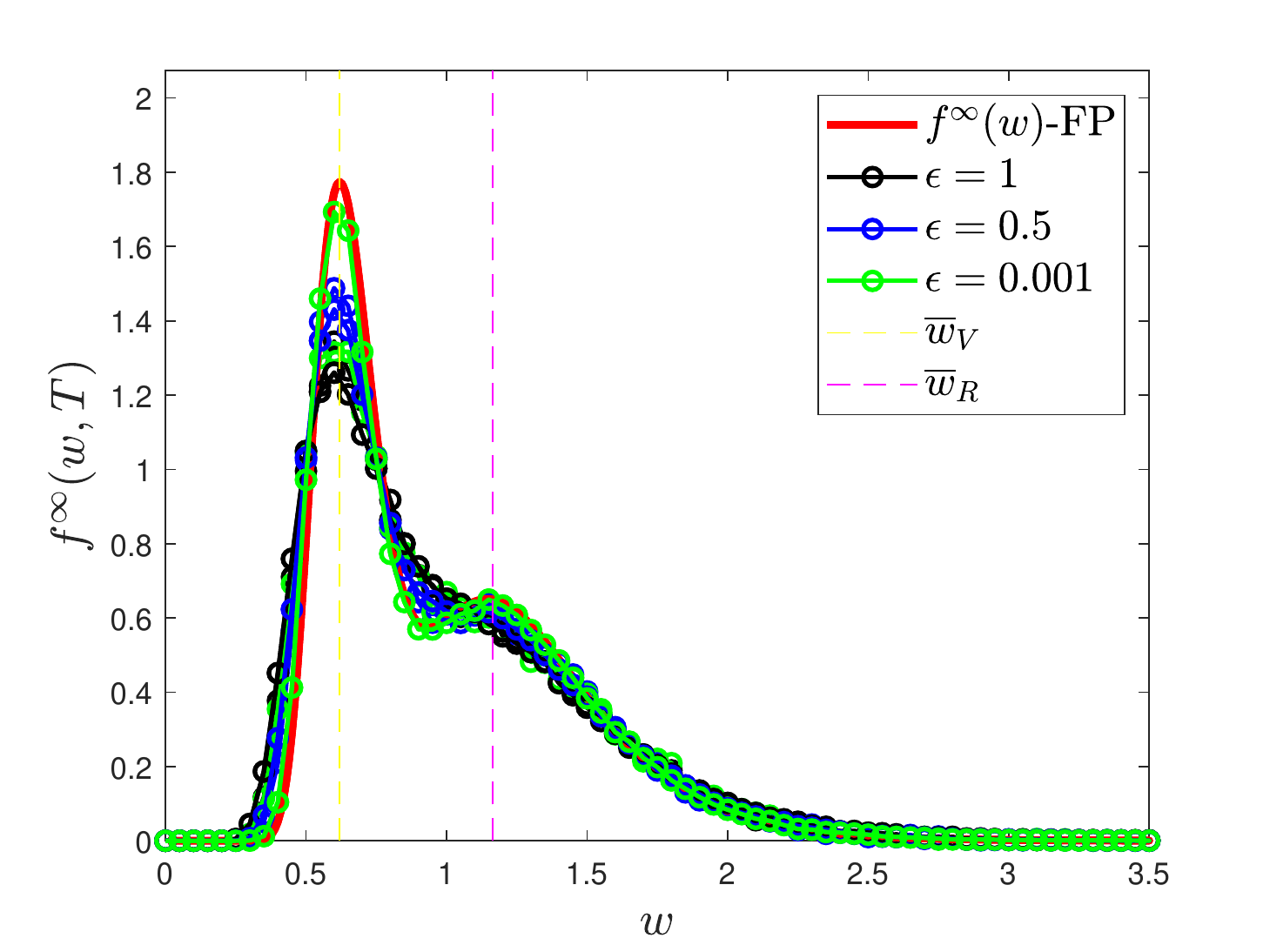}
    \caption{\textbf{Test 1}. Comparison of the wealth distributions at the end of the epidemic for the kinetic system \eqref{eq:kin_SIVR} with the explicit Fokker-Planck asymptotics \eqref{eq:global} with scaling parameters $\epsilon=1, \frac{1}{2}, 10^{-3}$. Left: $\lambda_S=0.15$, $\lambda_I=0.10$, $\lambda_S=0.30$, $\lambda_R=0.20$. Right: $\lambda_S=0.10$, $\lambda_I=0.05$, $\lambda_V=0.30$ $\lambda_R=0.15$. In both cases we fixed $\bar\beta = 0.2$, $\gamma_I = 1/12$, $\alpha = 0.005$, $\zeta=0.9$ and $\sigma^2 = 0.02$. }
\label{fig:1}
\end{figure}

\begin{figure}
    \centering
        \includegraphics[scale = 0.55]{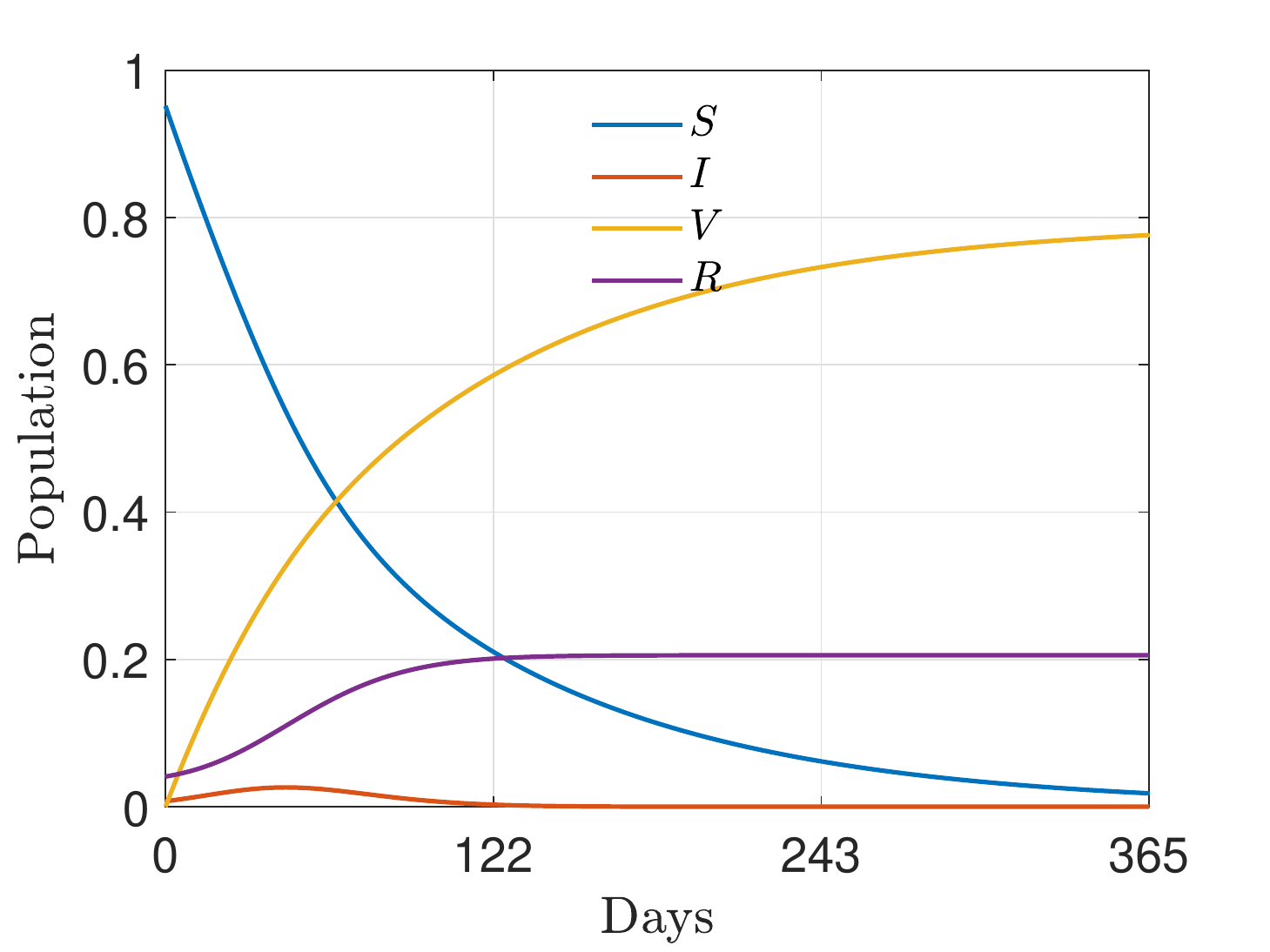}
    \includegraphics[scale = 0.55]{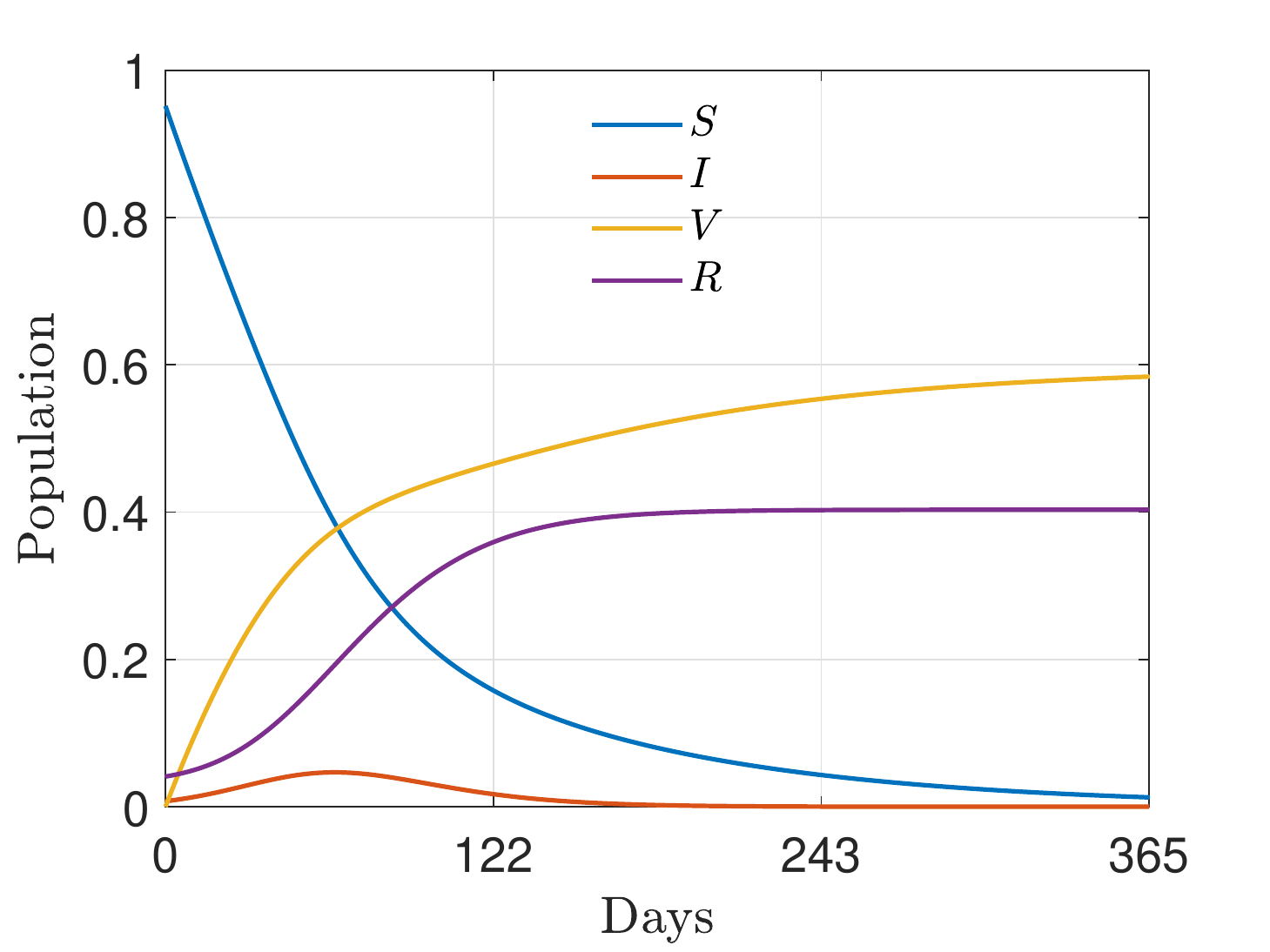}
    \caption{\textbf{Test 2}. Evolution of the epidemic dynamics from  \eqref{eq: SIVR} for the choice of parameters $\bar \beta = 0.15$, $\gamma_I = 1/12$, $\alpha = 0.01$ and $\zeta = 0.95$ (left) $\zeta = 0.55$ (right).  }
     \label{fig:2}
\end{figure}

\subsection{Test 2: Wealth inequalities and vaccination campaign}
\label{subsection: 4.3}
In the second test case we analyze the emergence of wealth inequalities through the computation of the Gini index. In particular, we concentrate on the effects linked to the outbreak of the infection and on the impact of an effective vaccination campaign. 


We fix the epidemic parameters as follows: $\bar\beta=0.15$, $\gamma_I=1/12$ and a vaccination rate $\alpha=10^{-2}$. Furthermore, we consider two different vaccine efficacies $\zeta = 0.95$ corresponding to a high efficacy of the vaccine, and $\zeta = 0.55$ corresponding to a low efficacy of the vaccine.  Since we are interested in the behavior of the system until the conclusion of the epidemic phenomenon,  the final time is fixed  as $T=810$, corresponding to a wide time span. 
We keep the same values for the saving propensities and market risk defined for Test \ref{test1}. Hence, we consider initial wealth distributions as in \eqref{eq:f0} and mass fractions as in \eqref{eq:fractions} with $\rho_I=7\times 10^{-3}$, $\rho_V=0$, $\rho_R=4\times10^{-2}$ and $\rho_S=1-(\rho_I+\rho_V+\rho_R)$. The scaling coefficient is $\epsilon=5\times 10^{-2}$. The resulting epidemic dynamic is reported in Figure \ref{fig:2}.

\begin{figure}
    \centering
        \includegraphics[scale = 0.55]{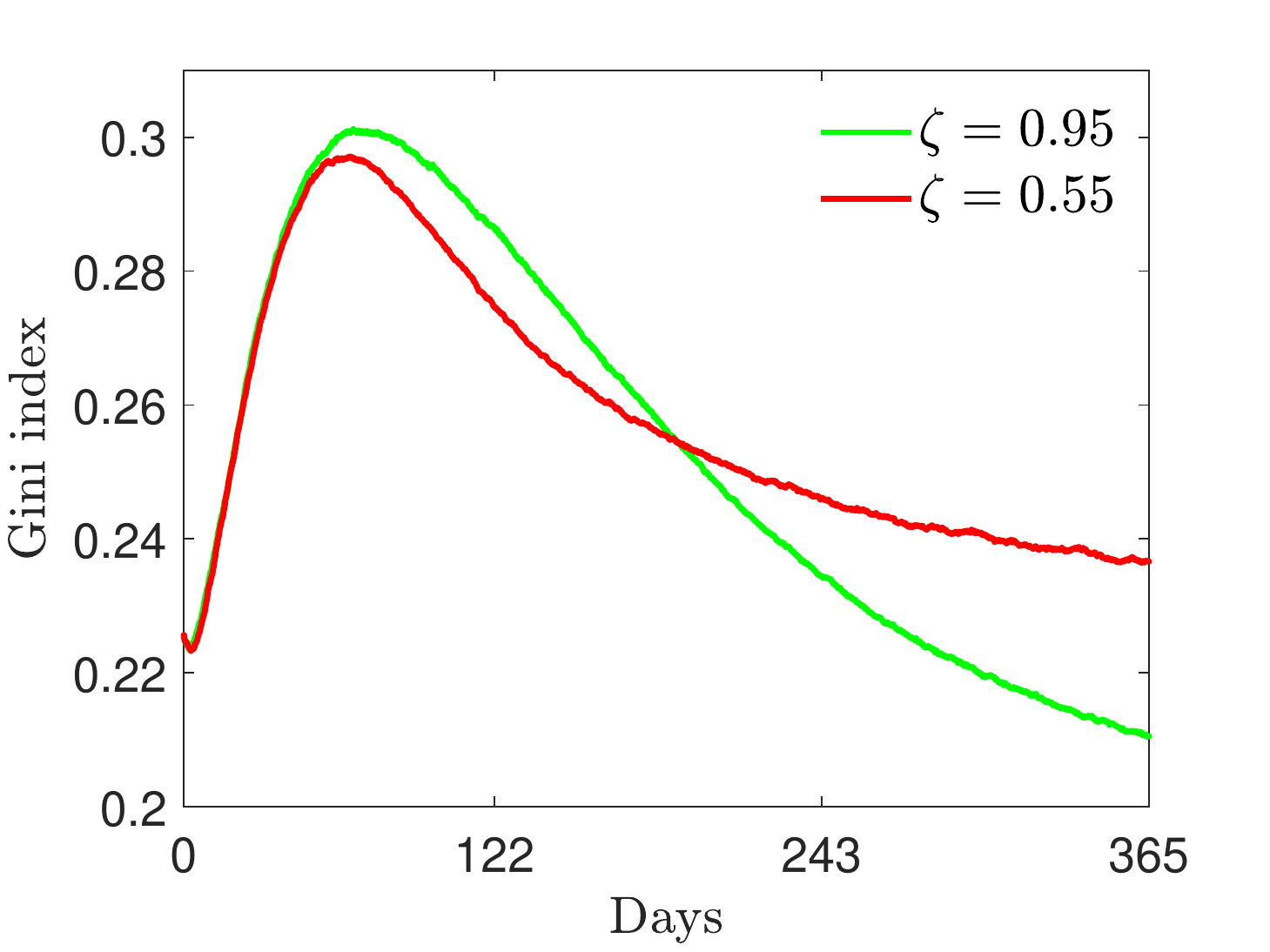}
          \caption{\textbf{Test 2}. Evolution of Gini index under the epidemic dynamics described in Figure \ref{fig:2} and for the choice of parameters  $\lambda_S=0.10$, $\lambda_I=0.07$, $\lambda_V=0.30$, $\lambda_R=0.15$. Two vaccine efficacies has been considered:  $95\%$ (green) and $55\%$ (red). In both cases we considered $\sigma^2=0.02$.}
     \label{fig:2b}
\end{figure}

We evaluate the Gini coefficient of the emerging equilibrium distributions.  The Gini index is commonly computed from the Lorenz curve
\[
L(F(w))=\int_{0}^{w}f^{\infty}(w_*)w_*dw_*,
\]
where $F(w)=\displaystyle\int_0^wf^{\infty}(w_*)dw_*$ and is defined as follows
\[
G_1=1-2\int_{0}^{1}L(x)dx.
\]
This index should be understood as a measure of a country's wealth discrepancy and it varies in $[0,1]$, where in the case $G_1=0$ the country is in a situation of perfect equality whereas $G_1=1$ means complete inequality. A reasonable value for this parameters is in the range $[0.2,0.5]$ for most western economies \cite{DPT}.

In Figure \ref{fig:2b} we show the evolution of the Gini index with the parameters described above. We may observe that the epidemic peak leads to an increasing of inequalities that is then absorbed for later times in relation to the efficacy of the vaccine. Consequently, only when the vaccine is made available to the majority of the population does it actually contribute to reducing inequalities, otherwise it may have the opposite effect. This reminds us of how, on a global level, the importance of making vaccines available to all countries should be seen not only in terms of epidemics, but also in terms of reducing economic inequalities.
In all the considered cases, in the long time, the Gini index decreases thanks to the vaccine.

Next, we consider the case where the market risk is related to the behaviour of the epidemic spreading and there is a linear relation between the market risk and the number of infected. The introduction of a time-dependent market risk $\sigma^2(t)$ mimics an instantaneous influence of the pandemic on the volatility of a market economy as often observed. Therefore, we consider the following 
\be
\label{eq:varsigma}
\sigma^2(t)=\sigma^2_0(1+\mu I(t))
\ee
where $\mu>0$ expresses the effective influence of the epidemic dynamics on the market volatility and $\sigma^2_0>0$ is an ineradicable baseline risk.

\begin{figure}
	\centering
		\includegraphics[scale = 0.55]{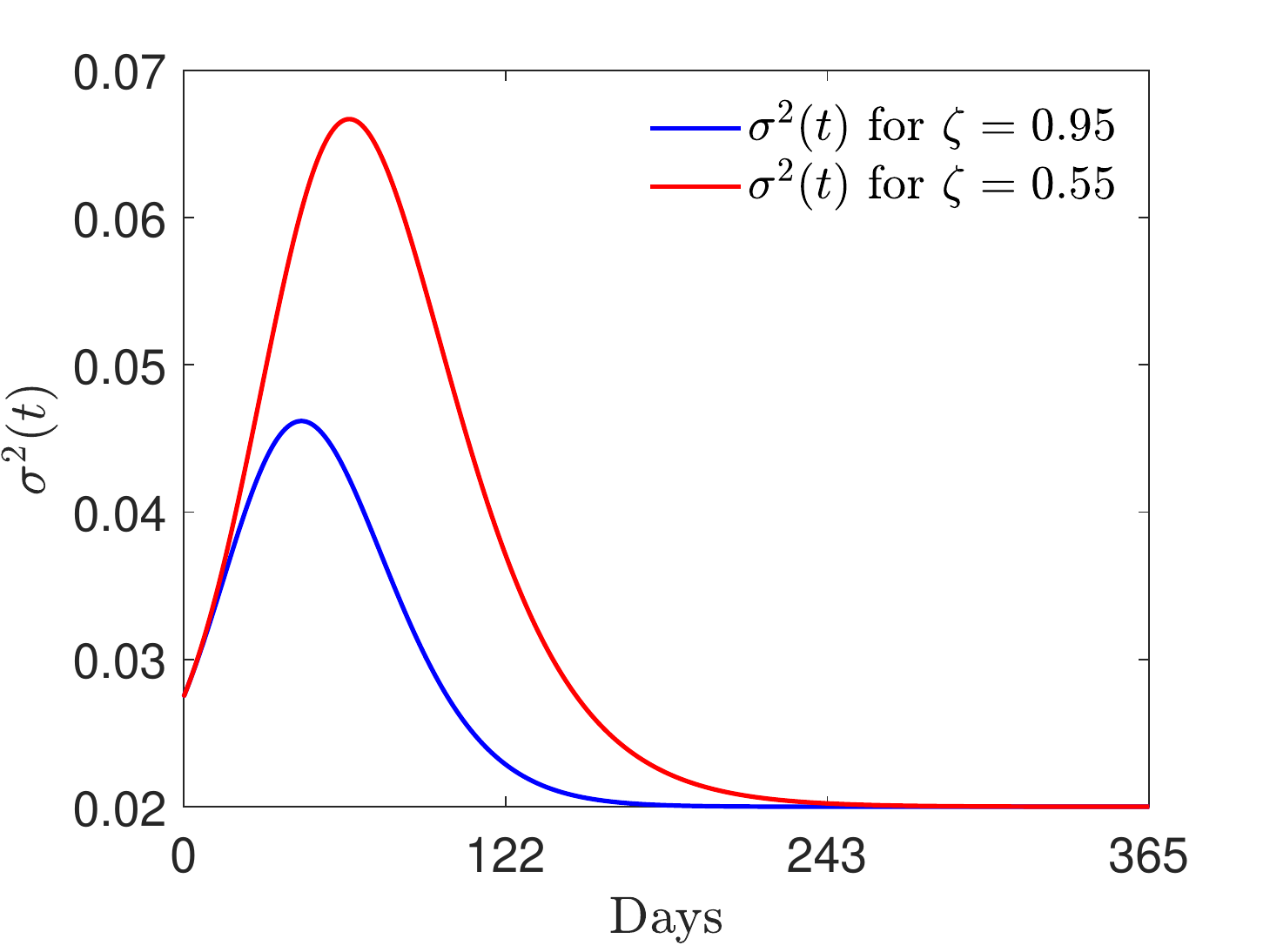}
	\includegraphics[scale = 0.55]{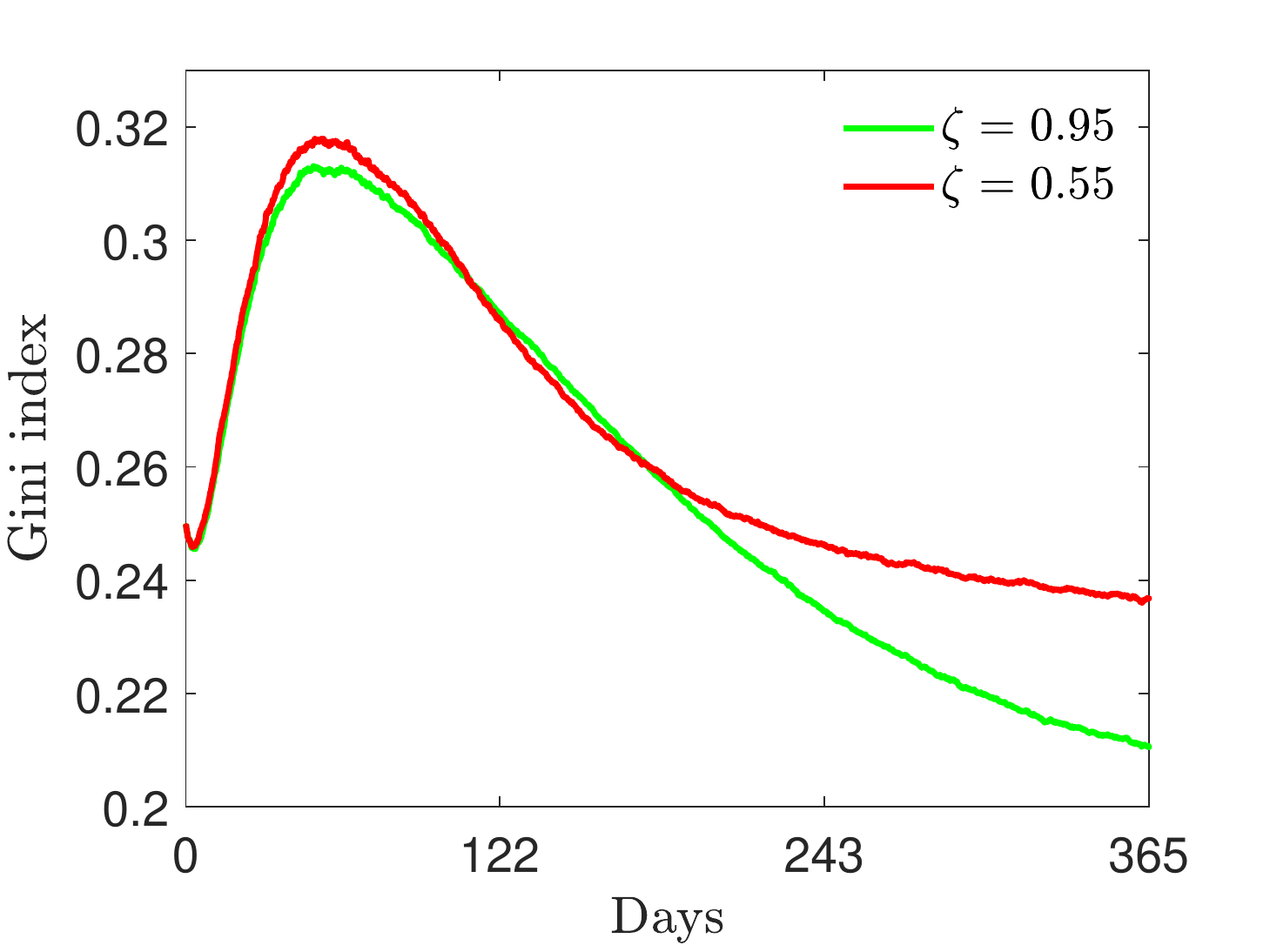}
	\caption{\textbf{Test 2}. Left: evolution of the market risk  $\sigma^2(t)$ as defined in \eqref{eq:varsigma} with $\mu=50$ and $\sigma_0^2=0.02$ in case of two different vaccine efficacy. Right: evolution of Gini index under the epidemic dynamics described in Figure \ref{fig:2} and epidemic-dependent market risk parameter \eqref{eq:varsigma}.  }
	\label{fig:2c}
\end{figure}

\begin{figure}
	\centering
		\includegraphics[scale = 0.55]{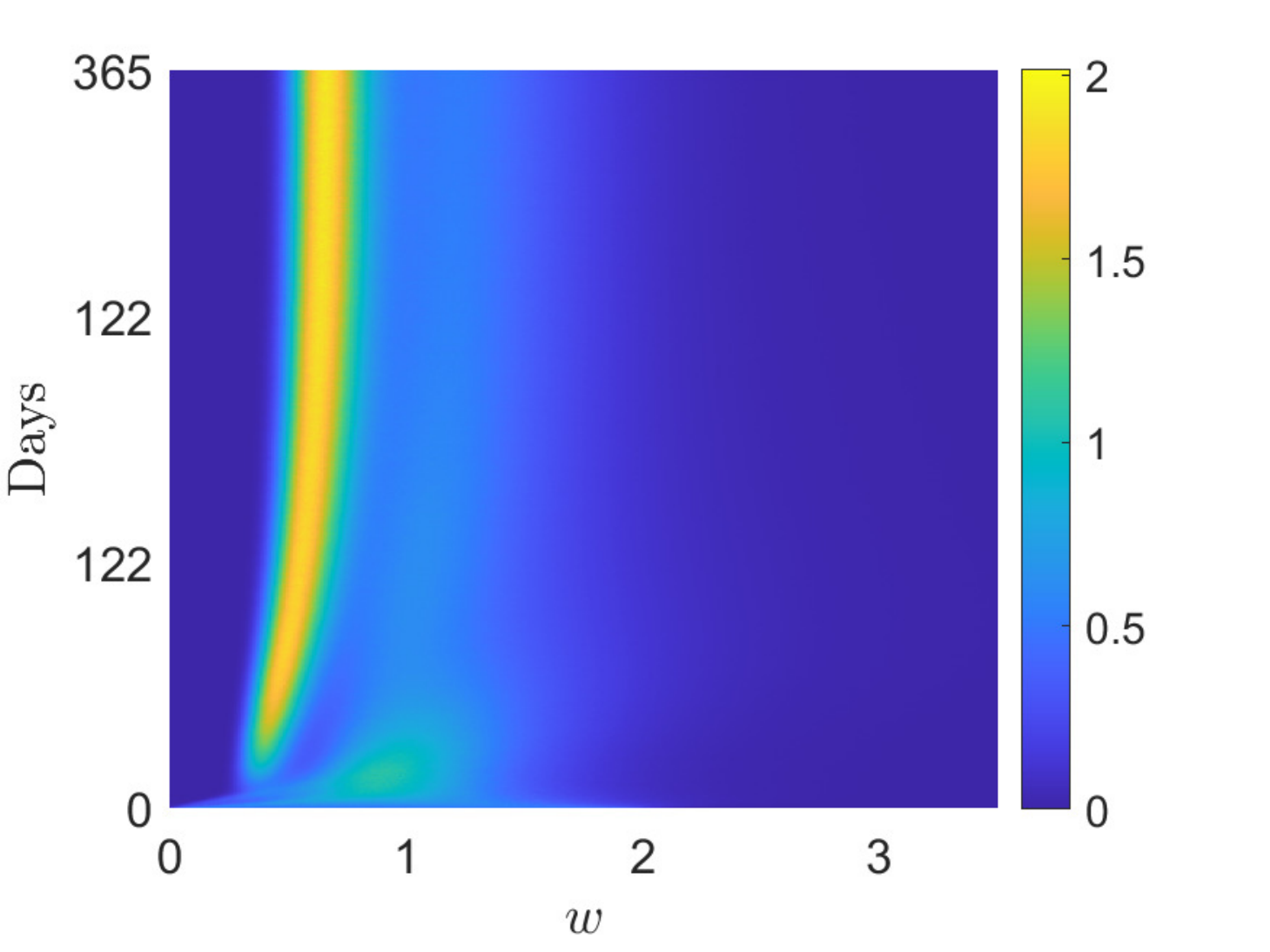}
		\includegraphics[scale = 0.55]{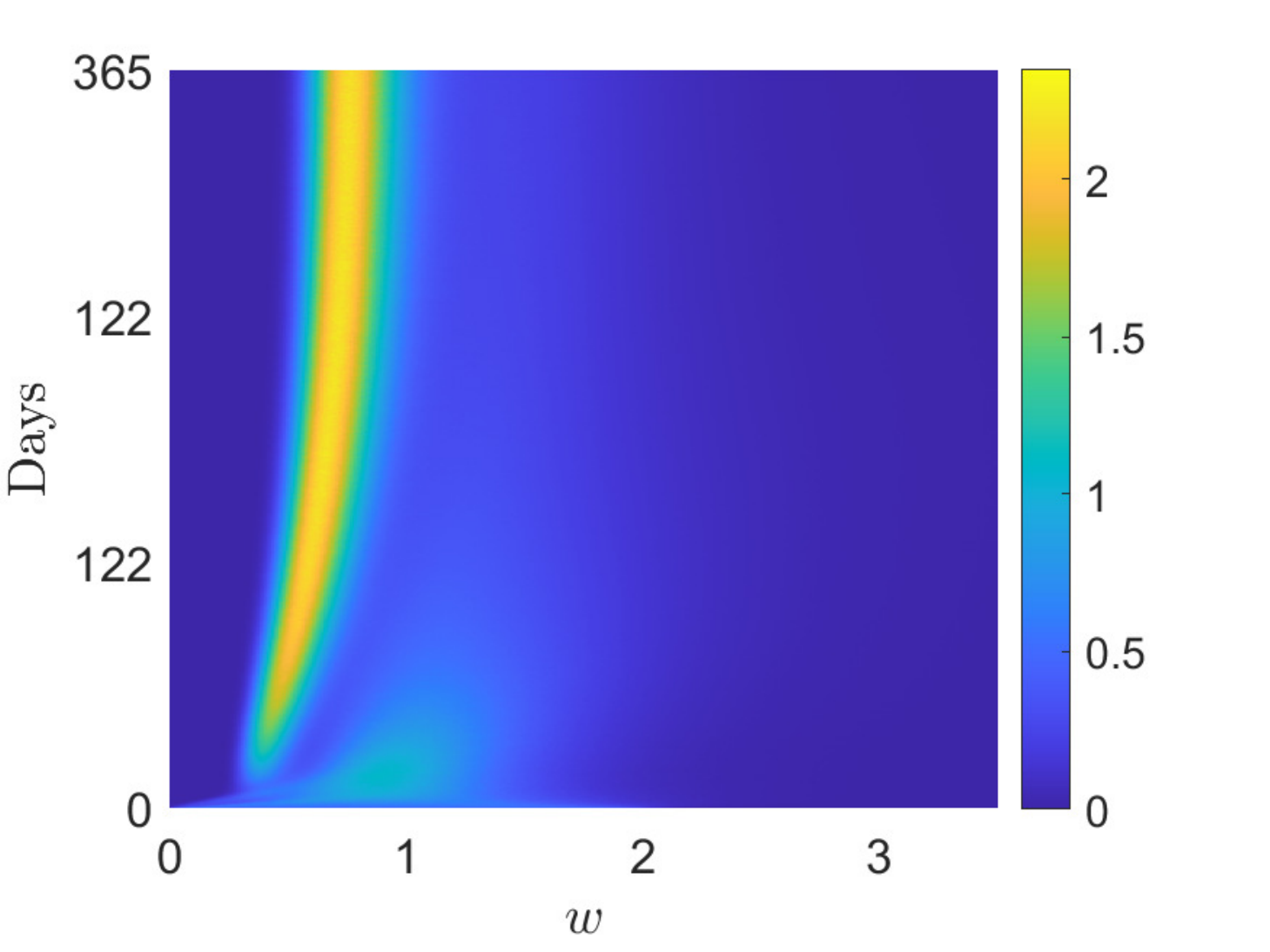}\\
		\includegraphics[scale = 0.55]{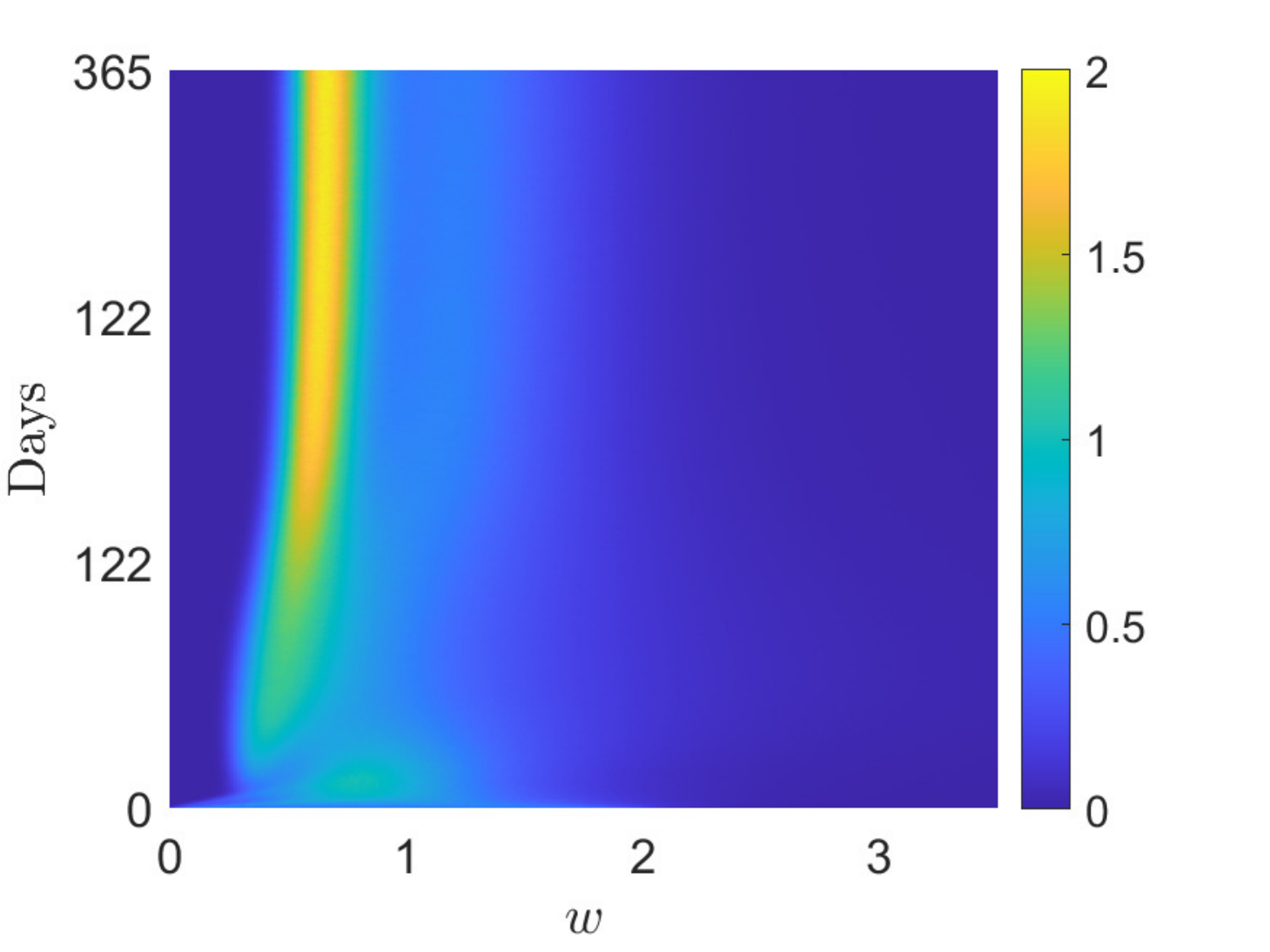}
		\includegraphics[scale = 0.55]{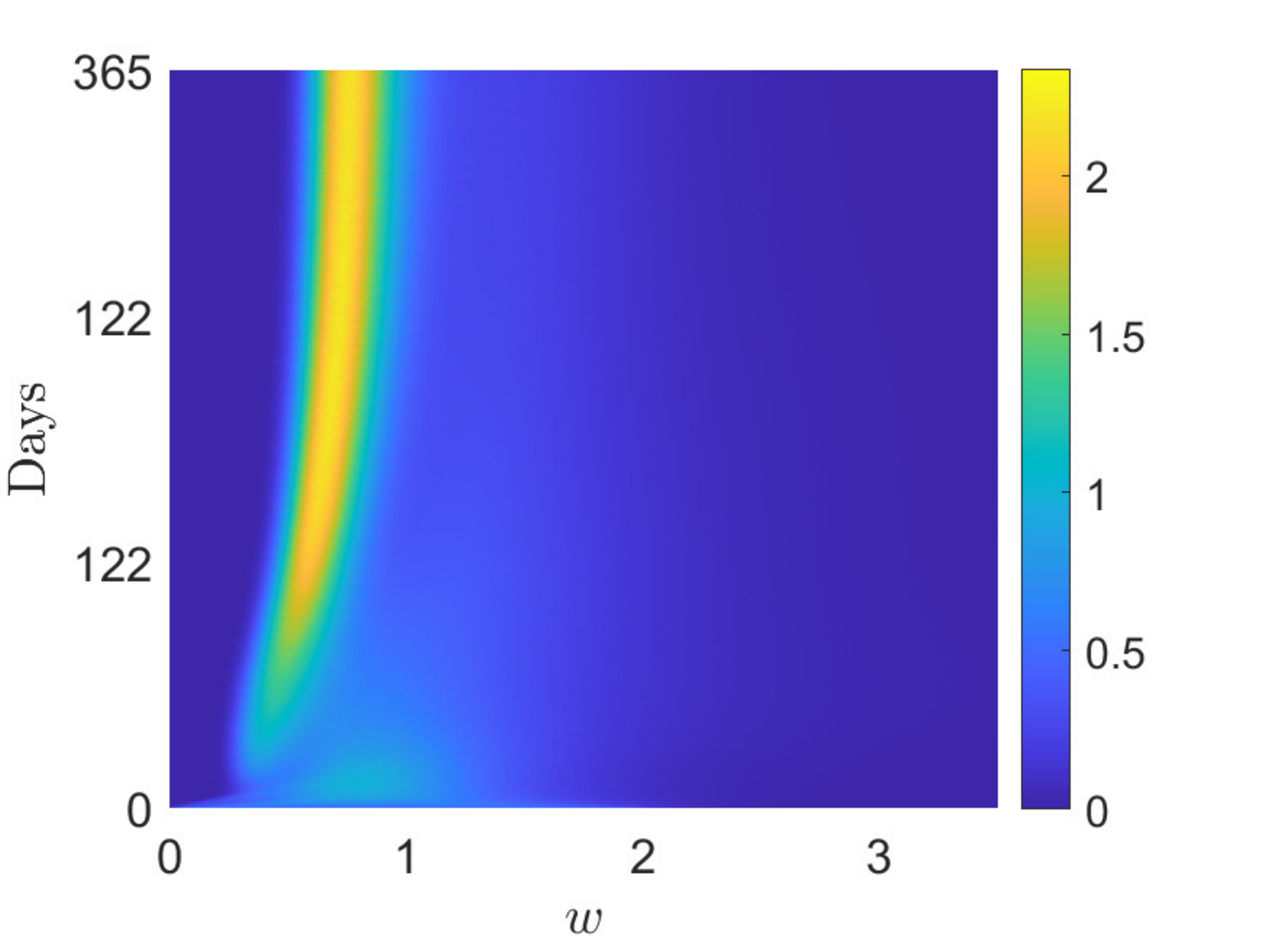}
	\caption{\textbf{Test 2}.  Time evolution of the wealth distribution of the kinetic model \eqref{eq:kin_SIVR} in the scaling $\epsilon = 5 \times 10^{-2}$ with vaccine efficacy $\zeta = 0.55$ (left column) or $\zeta = 0.95$ (right column) and with constant market risk $\sigma^2 = 0.02$ (top row) or $\sigma^2(t)$ defined in \eqref{eq:varsigma} with $\mu = 50$. In all the evolutions we considered $\lambda_S=0.10$, $\lambda_I=0.07$, $\lambda_V=0.30$ $\lambda_R=0.15$. The initial distribution has been defined in \eqref{eq:f0}-\eqref{eq:fractions}.
	 In the left image we have the evolution of the wealth distribution for the kinetic model \eqref{eq:kin_SIVR} in the scaling parameter $\epsilon=5\times10^{-2}$ with $\zeta=0.95$; whereas, in the right image we have the comparison between the behaviors of Gini index with vaccine effectiveness equal to $95\%$ (green line) and $65\%$ (red line). In both images we considered a variable market risk \eqref{eq:varsigma} with $\sigma_0^2=0.02$ and $\mu=50$ and $\lambda_S=0.10$, $\lambda_I=0.07$, $\lambda_V=0.30$ $\lambda_R=0.15$. }
	 \label{fig:6}  
\end{figure}

In the following we choose $\mu=50$ and $\sigma^2_0=0.02$. In Figure \ref{fig:2c} we represent the evolution of $\sigma^2(t)$ in presence of an epidemic characterized by $\bar\beta = 0.15$, $\gamma_I = 1/10$. Furthermore, we compare the Gini index in presence of two effectiveness of the vaccine, i.e. $\zeta=0.95$ and $\zeta=0.55$. We may easily observe how an increasing variability lead to a worsening of the Gini index and, therefore, of the inequalities. The large time behavior of the Gini index depends, as before, by the vaccine efficacy $\zeta$ such that low efficacy leads to increasing inequalities in the long time. This is due to the fact that as $t \rightarrow +\infty$ we have $I \rightarrow 0^+$ and then $\sigma^2(t) \rightarrow \sigma^2_0$.

Finally, in Figure \ref{fig:6} we present the evolution of the full kinetic density solution to \eqref{eq:kin_SIVR} in the scaling $\epsilon = 5 \times 10^{-2}$ in presence of fixed marked risk $\sigma^2$ or with the epidemic-dependent $\sigma^2(t)$ discussed in \eqref{eq:varsigma}. 

\subsection{Nonlinear incidence rate and time-varying vaccine efficacy}
In this last test case,  to model different frequency of interactions between agents that belong to the same wealth class, we introduce a wealth-dependent contact rate $\beta(w,w_*)$ of the form 
\be
\label{eq:nonlinearbeta}
\beta(w,w_*)=\frac{\bar\beta}{(c+|w-w_*|)^{\nu}},
\ee
where $\bar\beta$, $c, \nu>0$. We have depicted the above contact rate in Figure \ref{fig:beta}.
\begin{figure}
\centering
\includegraphics[scale = 0.55]{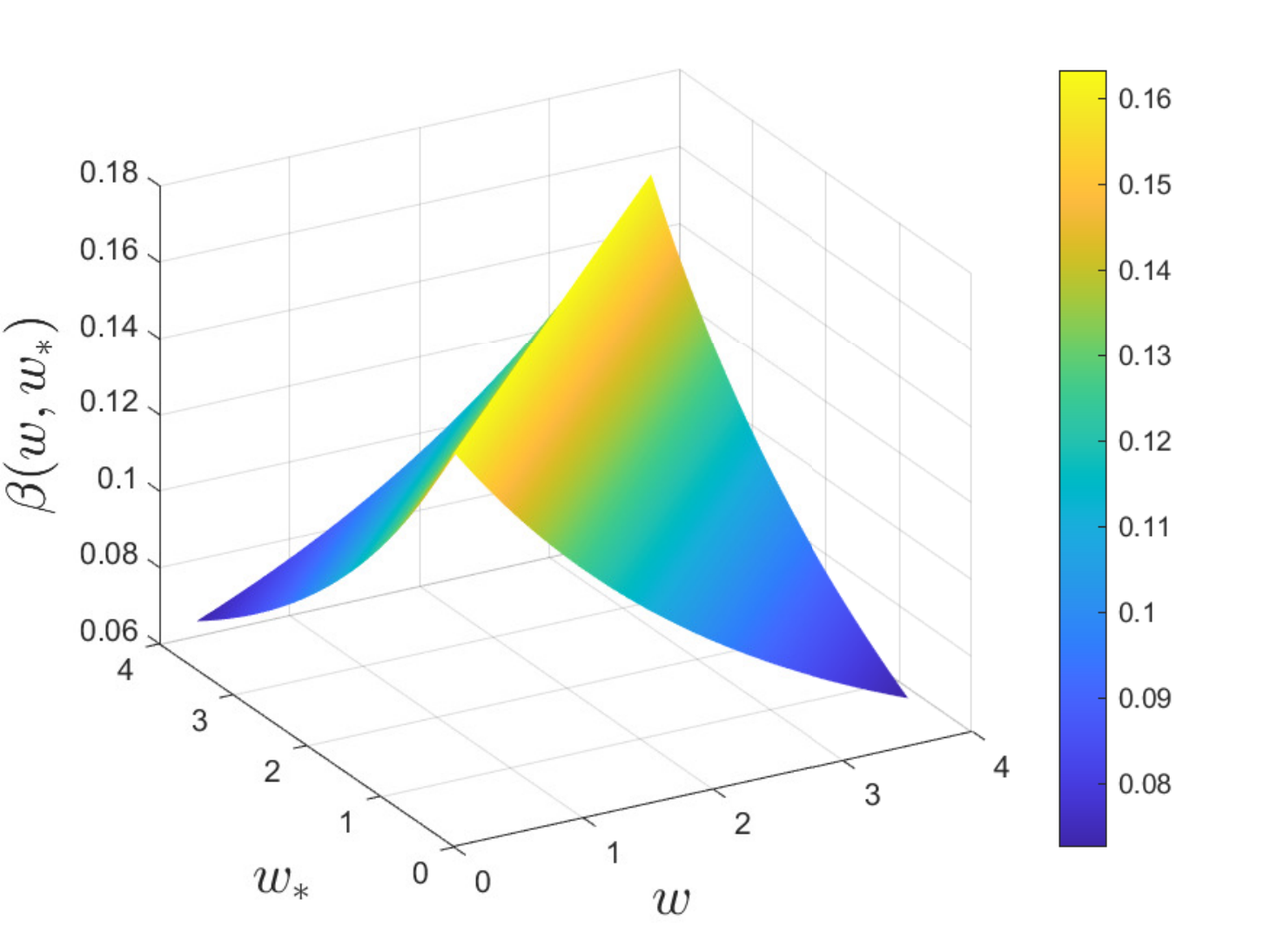}
\caption{\textbf{Test 3}. Wealth-dependent contact rate $\beta(w,w_*)$ of the form \eqref{eq:nonlinearbeta} with $\bar\beta=8$, $c=7$, $\nu=2$.  }
\label{fig:beta}
\end{figure}
We introduce also a time-dependent efficacy of the vaccine $\zeta$ of the form
\be
\label{eq:zeta}
\zeta(t)=\zeta_0-\psi \int_0^t \int_{\mathbb R_+} f_I(w,t)dwds = \zeta_0-\psi \int_0^t I(s)ds,
\ee
with $\zeta_0 \in [0,1]$ the initial efficacy of the vaccine and $0<\psi\le  \zeta_0$. This time dependence in vaccine coverage describes in a simplified way the fact that with more infected individuals it is more likely to encounter mutations of the original virus for which the vaccine is less effective. In the following, we compare the evolution of the wealth inequalities in presence of two different values $\zeta_0$. 

Furthermore, to make the modeling more realistic, we assume loss of immunity of the agents in the compartment $R$. To this end we have to modify the first and last equation of the model \eqref{eq:kin_SIVR} as follows
\begin{equation}
\label{eq:kin_re}
\begin{split}
\partial_t f_S(w,t) &= -K(f_S,f_I)(w,t) - \alpha f_S(w,t)  + \gamma_R f_R(w,t)+ \sum_{J \in \{S,I,V,R	\}} Q_{SJ}(f_S,f_J)(w,t) \\
\partial_t f_R(w,t) &= \gamma_I f_I(w,t)  - \gamma_R f_R(w,t)+ \sum_{J \in \{S,I,V,R	\}} Q_{RJ}(f_R,f_J)(w,t), 
\end{split}
\end{equation}
where $\gamma_R\ge 0$ is the rate expressing the loss of immunity of recovered agents. Note that  this last assumption substantially changes the epidemic dynamics, since asymptotically, instead of a disease free scenario, we will have the emergence of endemic states \cite{H, FP}.  

\begin{figure}
	\centering
	\includegraphics[scale = 0.55]{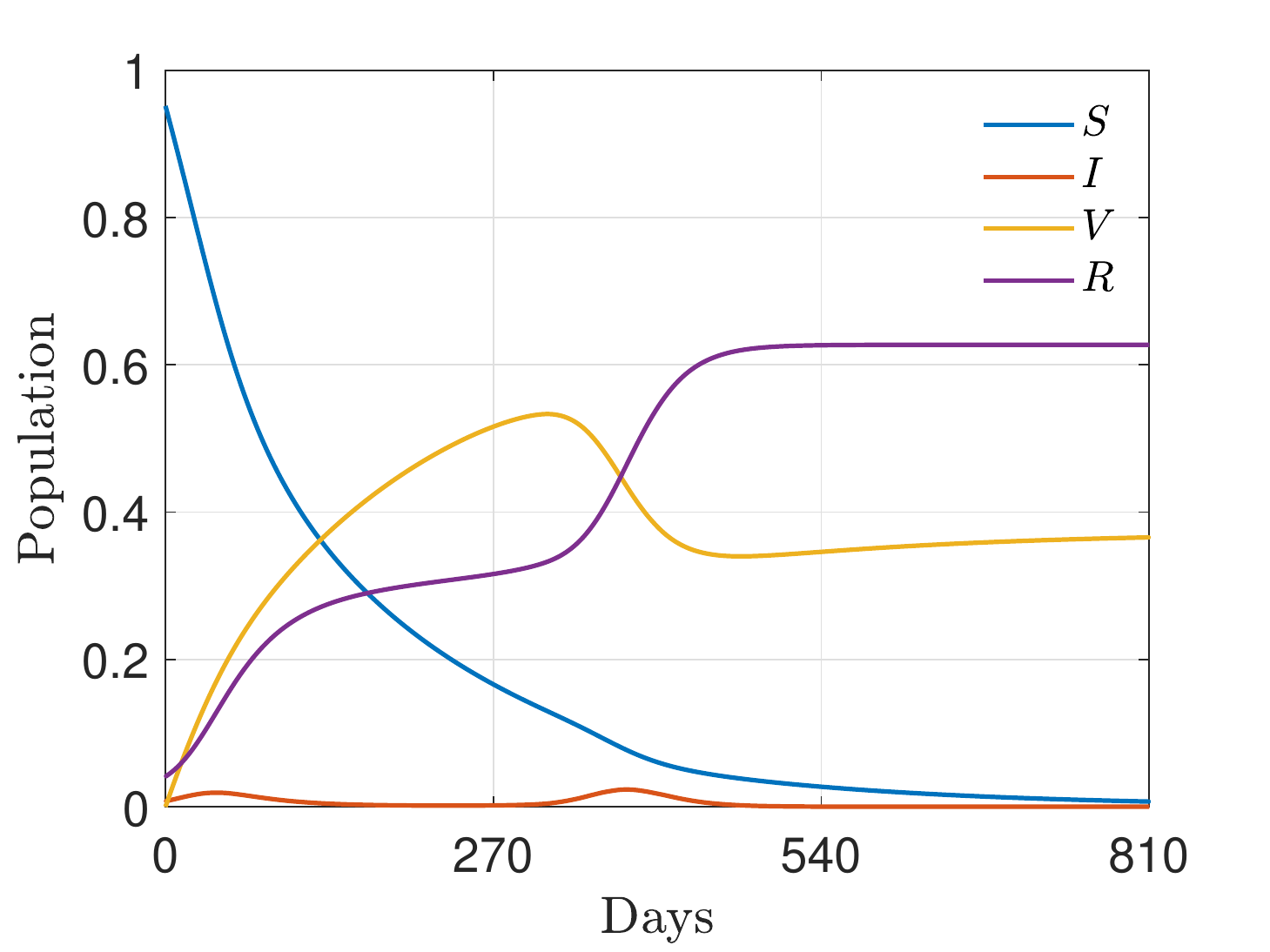}
	\includegraphics[scale = 0.55]{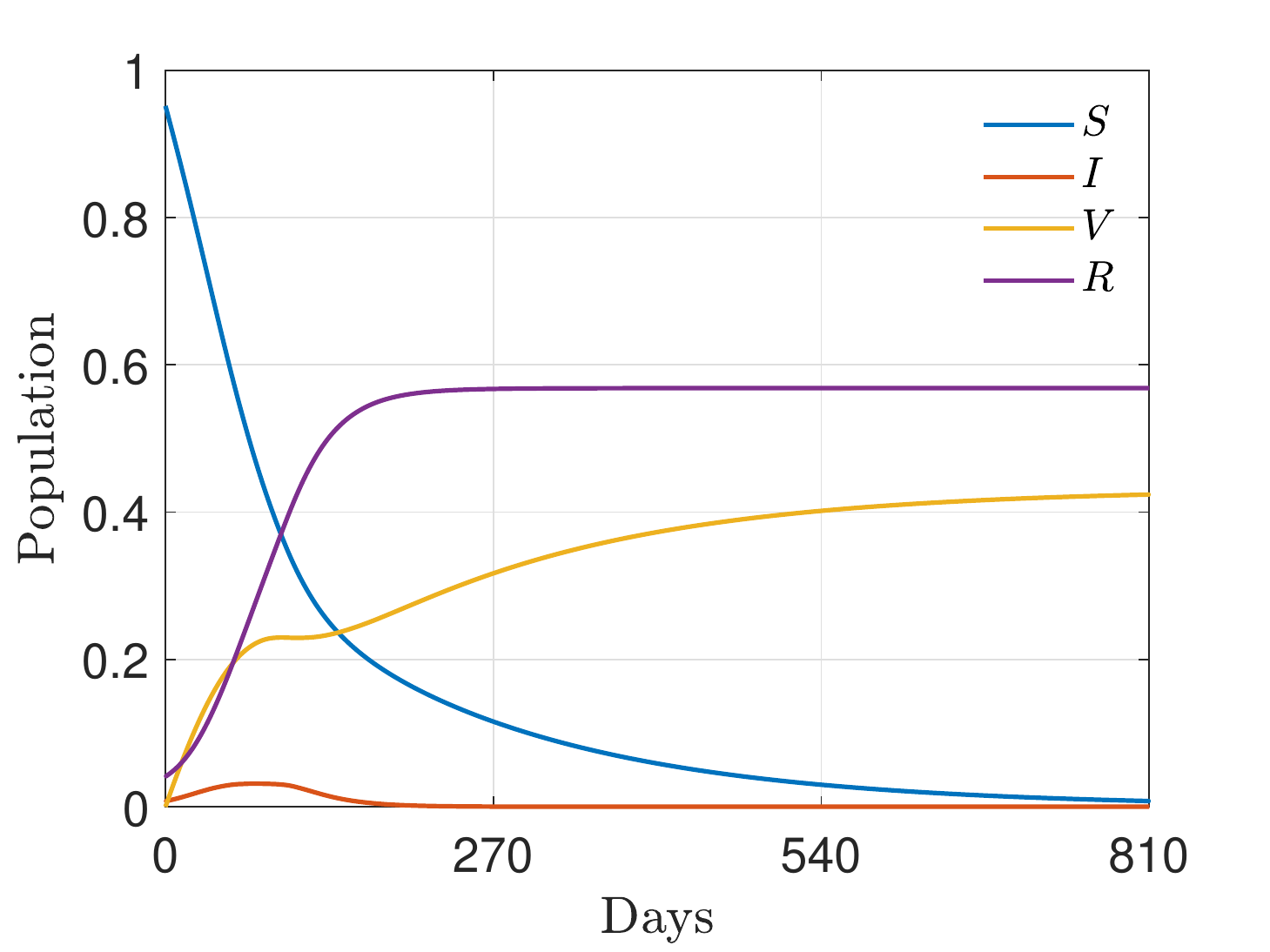}
	\includegraphics[scale = 0.55]{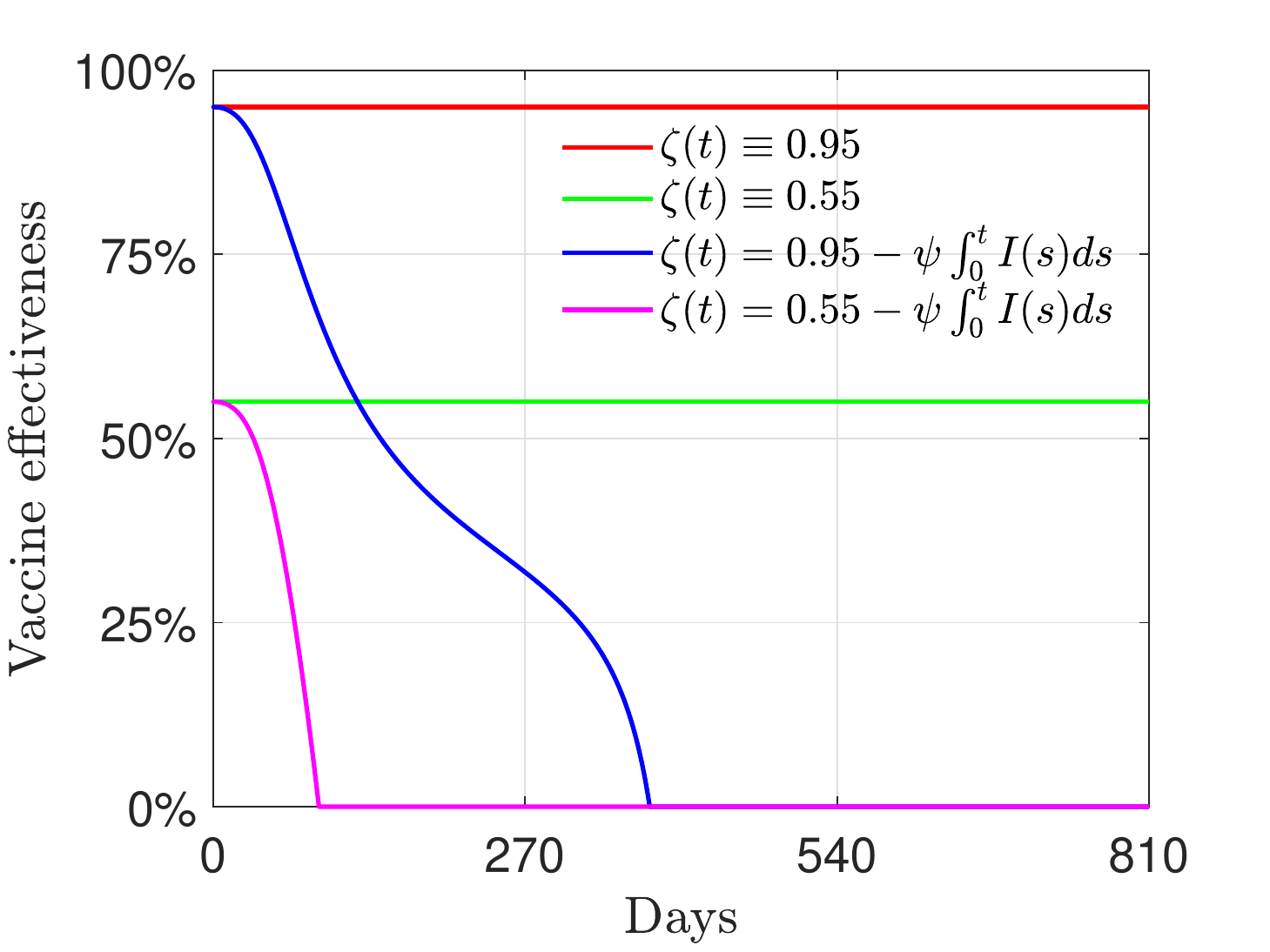}
	\includegraphics[scale = 0.55]{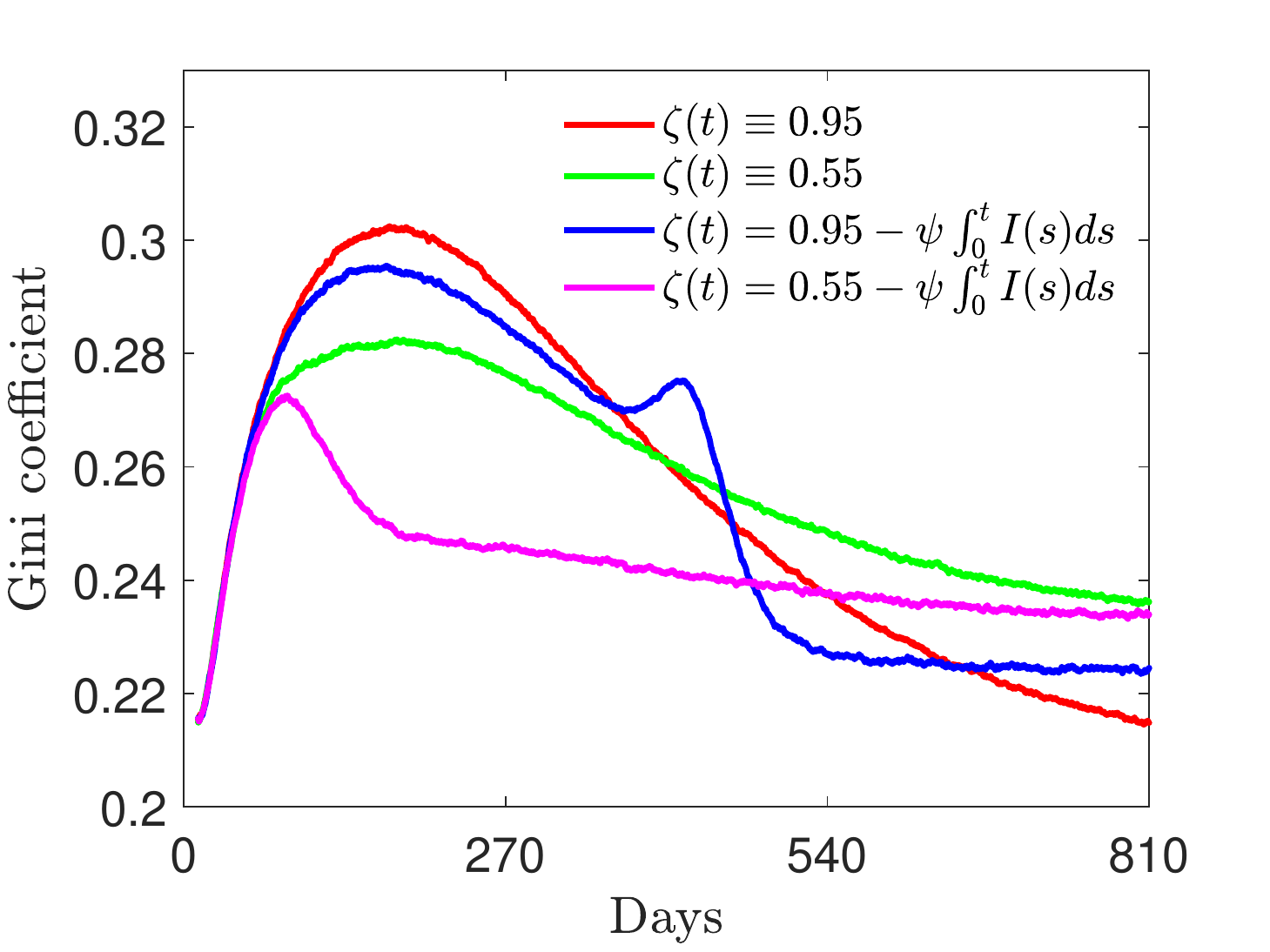}
	\caption{\textbf{Test 3A}. Top row: epidemic dynamics with wealth-dependent  $\beta(w,w_*)$ defined in \eqref{eq:nonlinearbeta} with  $\bar\beta=8$, $c=7$, $\nu=2$, $\gamma_I = 1/12$, $\alpha = 0.005$ and variable $\zeta$ as in \eqref{eq:zeta} with $\psi = 0.005$. We considered $\zeta_0 = 0.95$ (left) and $\zeta_0= 0.55$ (right). The initial distribution is \eqref{eq:f0} with mass fractions \eqref{eq:fractions}. Bottom row: decline of the vaccine efficacy due to the presence of high number of infectives (left) and evolution of the Gini index (right) for variable infection rate $\beta(w,w_*)$ as in \eqref{eq:nonlinearbeta} and vaccine effectiveness $\zeta(t)$ as in \eqref{eq:zeta}. We considered $\lambda_S=0.10$, $\lambda_I=0.07$, $\lambda_V=0.25$, $\lambda_R=0.15$ and $\bar\beta=8$, $c=7$, $\nu=2$, $\psi =0.005$.  }
	\label{fig:11}
\end{figure}

\subsubsection{Test 3A: $\gamma_R = 0$}

First, we consider model \eqref{eq:kin_SIVR} without the modified relations \eqref{eq:kin_re} (or equivalently, in absence of reinfection, i.e. $\gamma_R = 0$) and, as before, a fixed recovery rate $\gamma_I=1/12$ and vaccination rate $\alpha=0.005$ with the same initial masses defined in \eqref{eq:fractions}. Furthermore, we fixed $\psi = 0.005$. In Figure \ref{fig:11}, top row, we show the evolution for the fractions of the population in case of $\zeta_0=0.95$ (left) and $\zeta_0=0.55$ (right). We may observe how a variable efficacy of the vaccine that is affected by epidemic peaks may strongly shape the immunity of the population even in presence of initial high efficacy. Interestingly, in this latter case, a variable efficacy leads to the emergence of secondary peaks of infection. This is due to the presence of a smaller number of recovered persons who, unlike vaccinated persons, maintain immunity.



In Figure \ref{fig:11}, bottom-left row, we observe the evolution of the resulting vaccine efficacy for $\zeta_0 = 0.95$, $\zeta_0 = 0.55$ and $\psi = 0.005$. The vaccine efficacy is degraded by the epidemic dynamics due to the increasing of the infected compartment with a slower efficacy decay for high initial $\zeta_0$.  

For the same choice of coefficient, in the bottom-right plot of Figure \ref{fig:11}, we show the evolution of the Gini coefficient in the case of variable efficacy as \eqref{eq:zeta}. With respect to a vaccination with constant efficacy, the efficacy decay forces the emergence of sharper inequalities,  well evidenced by the evolution of Gini coefficient.  


\begin{figure}
\centering
\includegraphics[scale = 0.5]{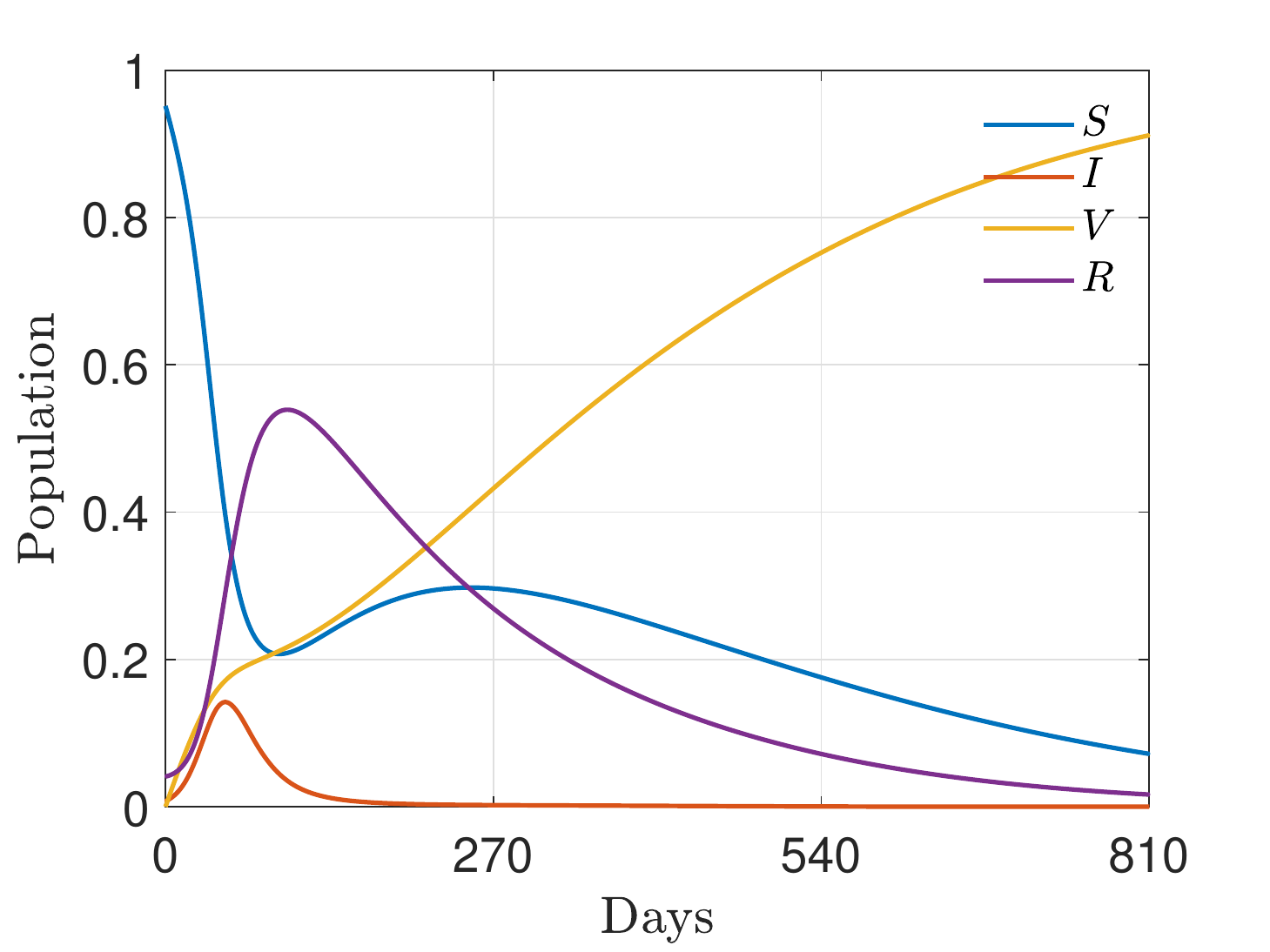}
\includegraphics[scale = 0.5]{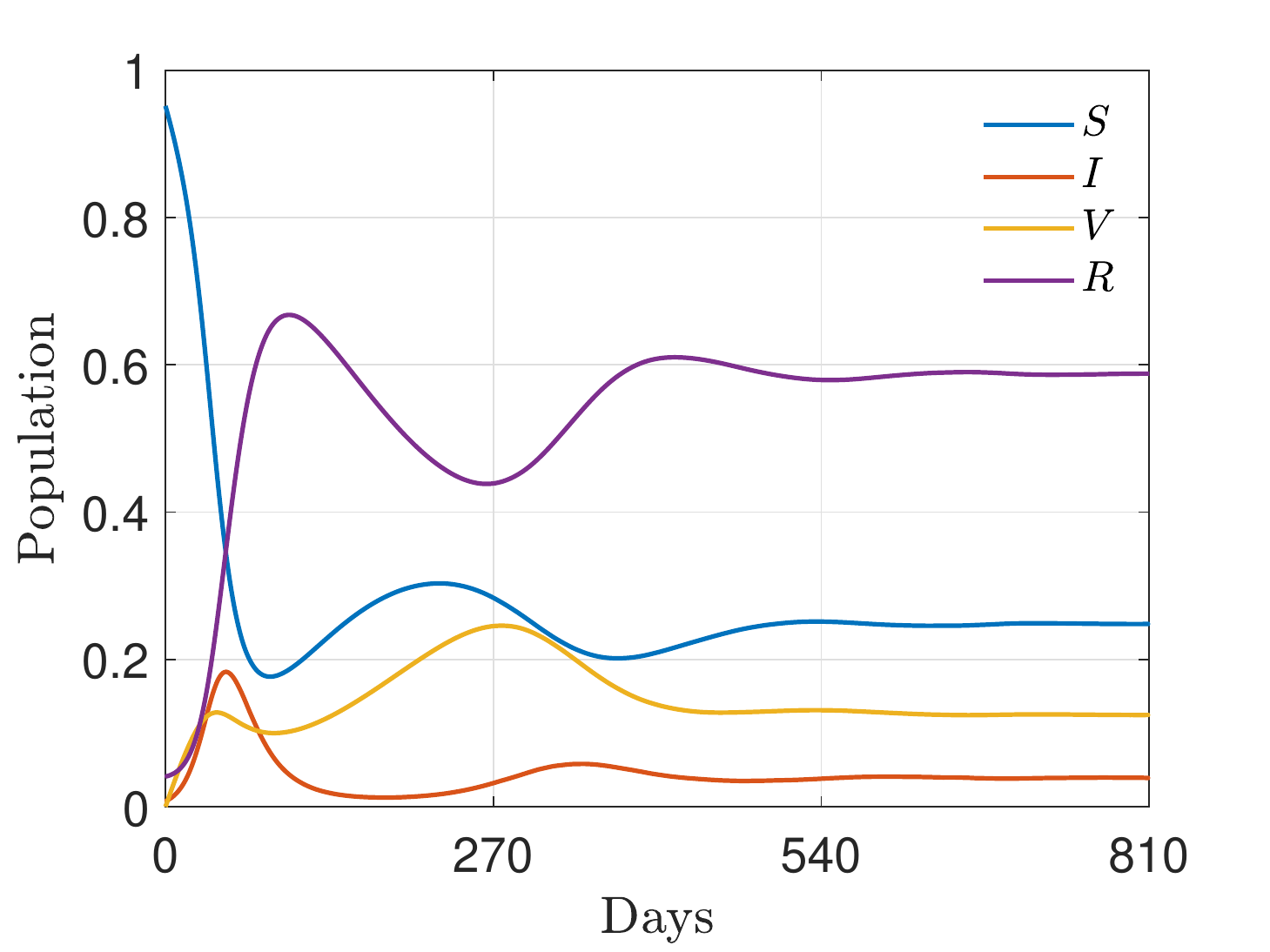} \\
\includegraphics[scale = 0.5]{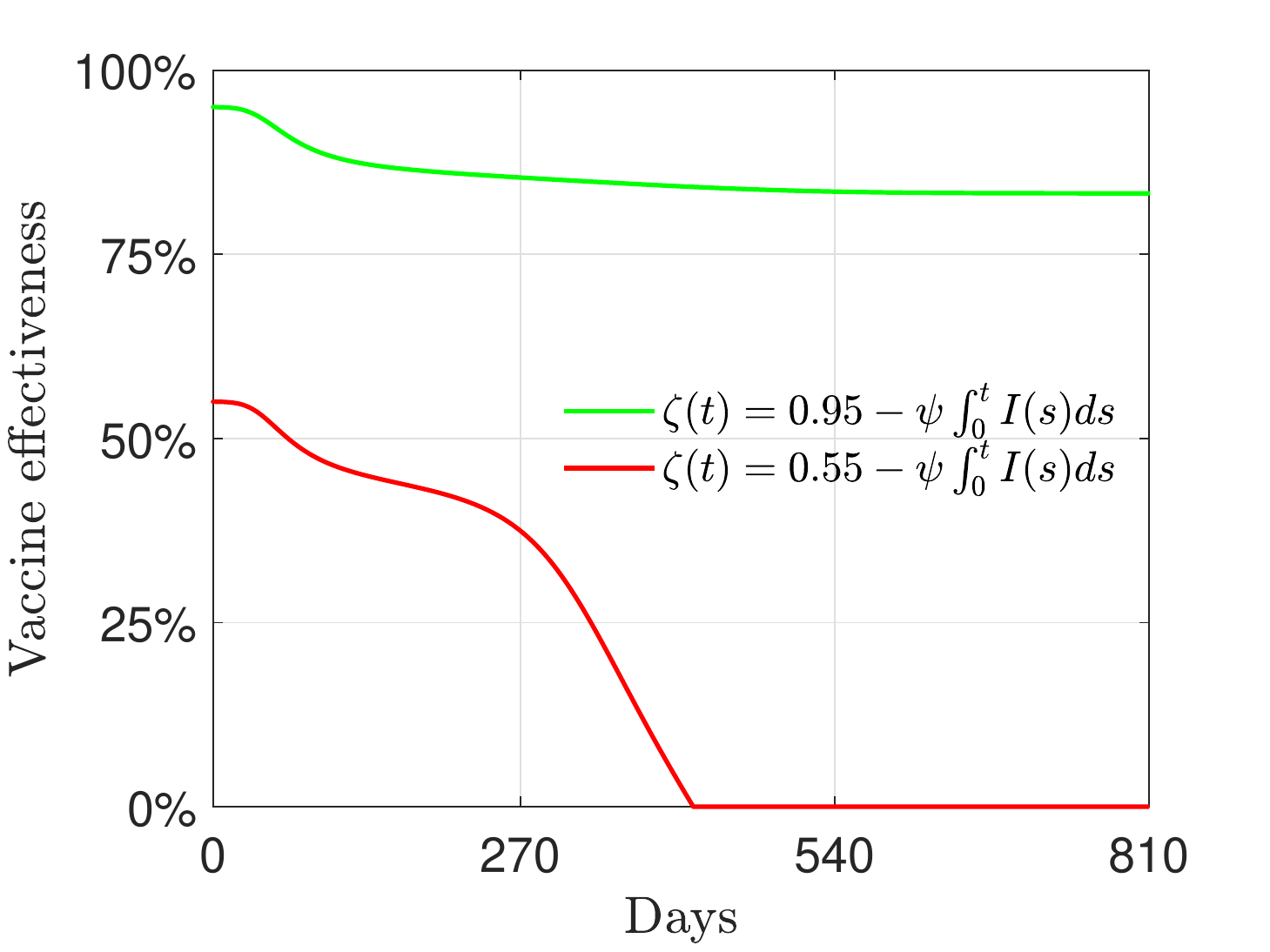}
\includegraphics[scale = 0.5]{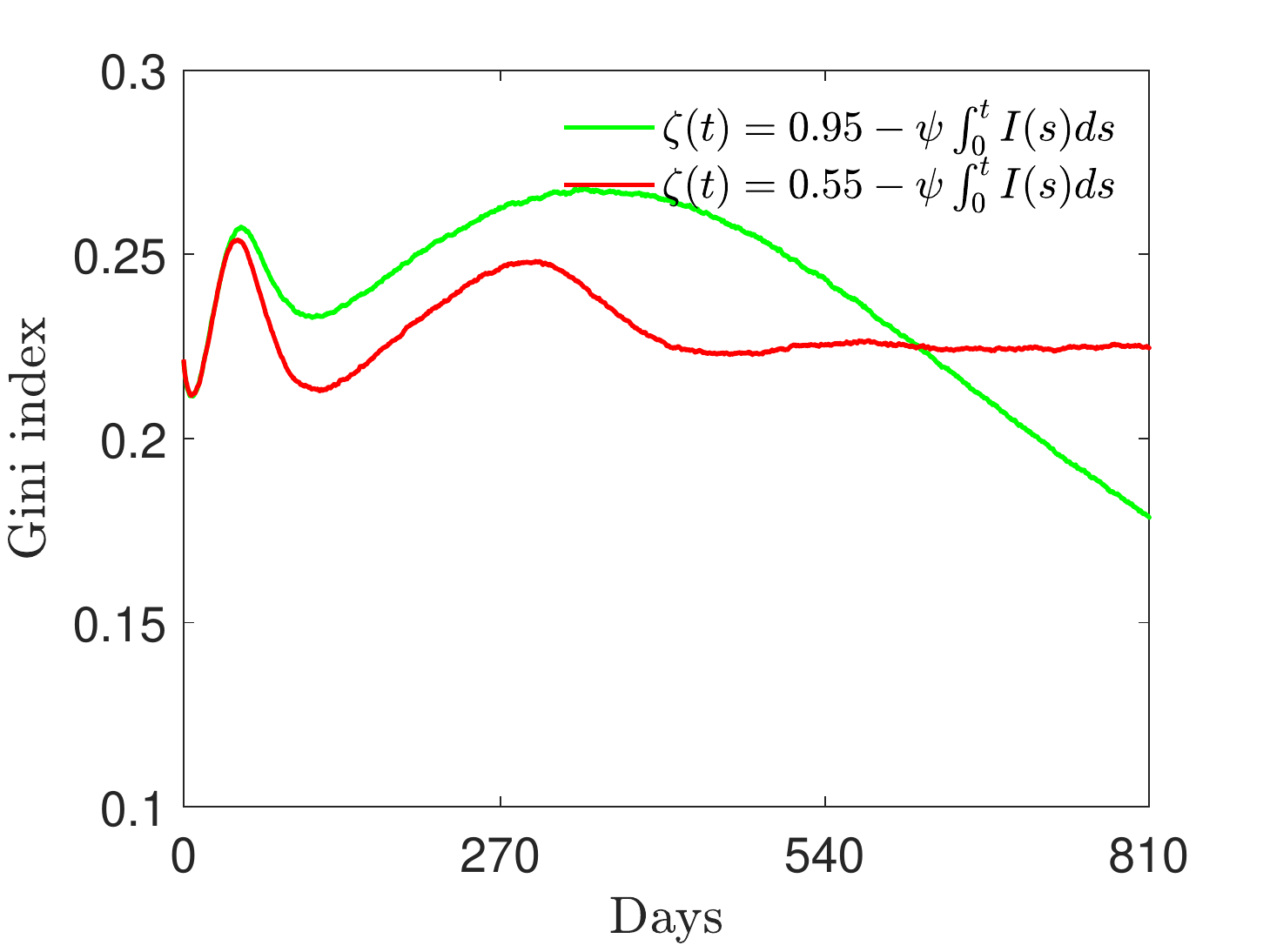}
\caption{\textbf{Test 3B}. Top row: epidemic dynamics with wealth-dependent $\beta(w,w_*)$ defined in \eqref{eq:nonlinearbeta} with  $\bar\beta=8$, $c=7$, $\nu=2$, $\gamma_I = 1/12$, $\gamma_R = 1/180$, $\alpha = 0.005$ and variable $\zeta$ as in \eqref{eq:zeta} with $\psi = 1.5\times 10^{-4}$. We considered $\zeta_0 = 0.95$ (left) and $\zeta_0= 0.55$ (right). The initial distribution is \eqref{eq:f0} with mass fractions \eqref{eq:fractions}.  Bottom row:  decline of the vaccine efficacy due to the presence of high number of infected (left) and evolution of the Gini index (right). We considered $\lambda_S=0.10$, $\lambda_I=0.07$, $\lambda_V=0.25$, $\lambda_R=0.15$ and $\bar\beta=8$, $c=7$, $\nu=2$. }
\label{fig:3B}
\end{figure}

\subsubsection{Test 3B: $\gamma_R >0$}

Finally, we consider model \eqref{eq:kin_SIVR} including the modified equations \eqref{eq:kin_re}, with a reinfection period of 180 days, i.e. $\gamma_R = 1/180$ and, as before, a fixed recovery rate $\gamma_I=1/12$ and vaccination rate $\alpha=0.005$ with the same initial masses defined in \eqref{eq:fractions}. In the first row of Figure \ref{fig:3B} we present two epidemic dynamics with nonlinear contact rate \eqref{eq:nonlinearbeta} and the  time dependent efficacy $\zeta(t)$ defined in \eqref{eq:zeta} with $\psi = 1.5 \times 10^{-4}$. In the left plot we present the case of strong initial vaccine efficacy $\zeta_0 = 0.95$ and in the right plot the case of mild initial vaccine efficacy $\zeta_0 = 0.45$. The macroscopic dynamics present an endemic equilibrium due to the presence of the reinfection rate $\gamma_R$. Also in contrast to the previous case, in the case of reduced initial efficacy of the vaccine,  a second infection wave is emerging.

Looking at the bottom-left plot we observe that, in the present regime of parameters, a strong initial vaccine efficacy is robust with respect to the efficacy decay due to epidemic waves. On the other hand, mild initial efficacies can dissipate their positive influence on the evolution of the infection. 
At the level of the evolution of the Gini index, in presence of reinfection, it appears even more evident that inequalities appear for large times in presence of mild vaccinations. Nevertheless, in the transient regimes,  the higher possibility to invest wealth for vaccinated agents, may create temporary inequalities.

\section*{Conclusions}
The widespread vaccination campaign undertaken in Western countries to counteract the evolution of the Covid-19 epidemic and its economic effects depends in large part on the efficacy of vaccines. Mathematical models capable of predicting the evolution of the economy in relation to the effectiveness of the vaccination campaign can play a fundamental role in configuring possible scenarios and suggesting further measures to be taken by governments. In this paper we analyzed, at the level of wealth distribution, the economic improvements induced by the vaccination campaign in terms of its percentage of effectiveness. Following the ideas developed in \cite{Albi_etal,DPTZ}, the interplay between the economic trend and the pandemic has been evaluated resorting to a mathematical model combining a kinetic model for wealth exchanges based on binary interactions with a classical SIR compartmental epidemic model including the compartment of vaccinated individuals. Although the model introduced necessarily represents a strong simplification of an extremely complex phenomenon, its qualitative behavior is capable of describing essential features of the pandemic's impact on individuals' wealth. A key aspect of the model is, in fact, the possibility of obtaining explicit configurations of the stationary wealth distributions in the form of inverse Gamma densities, with the essential parameters depending on the percentage of vaccinated and recovered individuals, thus relating the effectiveness of the vaccination campaign to the formation of wealth inequalities. Several numerical experiments have also been conducted to quantify how a highly effective vaccination campaign has a direct effect on the decrease over time of the Gini coefficient, a classic measure of inequality in the distribution of wealth in Western societies. As a concluding remark we highlight how the approach introduced here based on the combination of classical compartmental models with wealth exchange processes lends itself naturally to numerous generalizations, both in terms of economic interactions and in terms of epidemic dynamics.

\section*{Acknowledgment}
This work has been written within the activities of the GNFM and GNCS groups of INdAM (National Institute of High Mathematics). M.Z. acknowledge partial support of MUR-PRIN2020 Project "Integrated mathematical approaches to socio-epidemiological dynamics". The research of M.Z. was partially supported by MIUR, Dipartimenti di Eccellenza Program (2018–2022), and Department of Mathematics “F. Casorati”, University of Pavia. The research of L.P. was partially supported by FIR2021 project “No hesitation. For effective communication of Covid-19 vaccination”, University of Ferrara.

\end{document}